\documentclass[sigconf]{acmart}
\usepackage{amsmath}
\usepackage{amsthm}
\usepackage{bm}
\usepackage{booktabs,graphicx,color,xspace, mathtools, multirow, listings, relsize, dirtree,todonotes,comment,amssymb,stmaryrd,tikz,subcaption}

\usepackage[setrelation]{math}
\usepackage[prefixflatinterpret,bracketmodalinterpret,fixformat,silentconst,sidenotecalculus,longseqcontext]{logic}
\usepackage[prefixflatinterpret,bracketmodalinterpret,fixformat,silentconst,differentialdL,simplenames]{dL}

\setcopyright{rightsretained}

\acmDOI{10.475/123_4}

\acmISBN{123-4567-24-567/08/06}

\acmConference[ICCPS'19]{ACM/IEEE International Conference on Cyber-Physical Systems}{April 2019}{Montreal, Canada}
\acmYear{2019}
\copyrightyear{2019}

\lstdefinelanguage{ST}{%
  keywords={IF,THEN,ELSE,END_IF,VAR_INPUT,VAR_OUTPUT,END_VAR,REAL,
  						BOOL,PROGRAM,END_PROGRAM,CONFIGURATION,RESOURCE,TASK,
            INTERVAL,END_RESOURCE,END_CONFIGURATION,WITH,PRIORITY,
            ON,PLC,OR,AND},%
  sensitive=true,
  morecomment=[s]{/*}{*/},
  deletestring=[d]',
  showstringspaces=false,
  commentstyle=\color{vgreen},
  mathescape,
  escapeinside={/*@}{@*/}}[keywords]

\begin{document}

\newtheorem{remark}{Remark}
\newtheorem{case}{Case}
\newtheorem{ih}{IH}

\newcommand{\name}{\textsc{HyPLC}\xspace}
\newcommand{\nametitle}{HyPLC\xspace}
\newcommand{\etal}{\xspace \emph{et al.}\xspace}

\newcommand{\compilest}[1]{\textsf{\small HP}$(#1)$}
\newcommand{\compilehp}[1]{\textsf{\small ST}$(#1)$}
\newcommand{\compilesteq}[1]{\textsf{\small HP}(#1)}
\newcommand{\compilehpeq}[1]{\textsf{\small ST}(#1)}
\newcommand{\compileop}{\ensuremath{\triangleright}}

\newcommand{\stprogram}[1]{$s_{#1}$}
\newcommand{\stprogrameq}[1]{s_{#1}}

\newcommand{\skipp}{\textbf{skip;}\xspace}
\newcommand{\debug}[1]{}
\newcommand{\says}[3]{\debug{\todo[size=\small,color=#2,inline]{#1 says: #3}}}

\newcommand{\luis}[1]{\says{Luis}{pink}{#1}}
\newcommand{\andre}[1]{\says{Andr\'{e}}{green}{#1}}
\newcommand{\stefan}[1]{\says{Stefan}{yellow}{#1}}

\newcommand{\orange}[1]{\textcolor{orange}{#1}}
\newcommand{\green}[1]{\textcolor{orange}{#1}}
\newcommand{\blue}[1]{\textcolor{blue}{#1}}
\newcommand{\red}[1]{\textcolor{red}{#1}}

\def\dl{\dL}
\def\kyx{\KeYmaeraX}

\newcommand{\arithmeticOp}{\ensuremath{\sim}}
\newcommand{\relationalOp}{\ensuremath{\bowtie}}
\newcommand{\logicalOp}{\ensuremath{\frown}}
\newcommand{\relationalOpST}{\ensuremath{\relationalOp_{\texttt{ST}}}}
\newcommand{\logicalOpST}{\ensuremath{\logicalOp_{\texttt{ST}}}}
\newcommand{\relationalOpHP}{\ensuremath{\relationalOp_{\texttt{HP}}}}
\newcommand{\logicalOpHP}{\ensuremath{\logicalOp_{\texttt{HP}}}}

\providecommand{\ivr}{Q}%
\newcommand{\ws}{\nu}%
\newcommand{\wt}{\omega}%
\newcommand{\wsig}{\mu}%
\newcommand{\wst}{\sigma}%
\newcommand{\stdI}{\interpretation[state=\ws]}%
\newcommand{\I}{\iconcat[state=\ws]{\stdI}}%
\newcommand{\It}{\iconcat[state=\wt]{\stdI}}%
\newcommand{\Isig}{\iconcat[state=\wsig]{\stdI}}%
\newcommand{\Ist}{\iconcat[state=\wst]{\stdI}}%

\newcommand{\stskip}{\textsf{skip}}
\newcommand{\stopeval}{\ensuremath{\rightarrow_a}}
\newcommand{\stprgeval}{\ensuremath{\rightarrow}}
\newcommand{\stprg}[1]{\ensuremath{#1}}
\newcommand{\stassign}[2]{\ensuremath{#1 := #2}}
\newcommand{\stif}[3]{\ensuremath{\textsf{if}~(#1)~\textsf{then}~#2~\textsf{else}~#3}}
\newcommand{\stifthen}[2]{\ensuremath{\textsf{if}~(#1)~\textsf{then}~#2}}
\newcommand{\stsubst}[3]{\ensuremath{#1 [#2 \mapsto #3]}}

\newcommand{\ctrl}{u {\,:\in\,} \text{ctrl}(x,i)}
\newcommand{\extinput}{\prandom{i}}
\newcommand{\scduration}{\ensuremath{\varepsilon}}
\newcommand{\plant}{\ensuremath{\humod{t}{0};\{\D{x}=f(x,u),\D{t}=1~\&~t{\leq}\scduration\}}}
\newcommand{\plantabbrv}{\text{plant}}
\newcommand{\hpscancycle}{\ensuremath{\extinput;{\ctrl};\plant}}
\newcommand{\hpscancycleabbrv}{\ensuremath{\extinput;{\ctrl};\plantabbrv}}
\newcommand{\stscancycle}{\ensuremath{\extinput;\compilehp{\ctrl};\plant}}

\newcommand{\boxalign}[2][0.4\textwidth]{
  \par\noindent\tikzstyle{mybox} = [draw=black,inner sep=8pt]
  \begin{center}\begin{tikzpicture}
   \node [mybox] (box){%
    \begin{minipage}{#1}{\vspace{-5mm}#2}\end{minipage}
   };
  \end{tikzpicture}\end{center}
}

\title[\nametitle: Hybrid PLC Translation for Verification]{\nametitle: Hybrid Programmable Logic Controller \\Program Translation for Verification
}
\author{Luis Garcia}
\affiliation{%
  \department{\textit{Electrical and Computer Engineering Department}}
  \institution{University of California, Los Angeles}
  \city{Los Angeles}
  \state{CA}
  \country{USA}
  }
\email{garcialuis@ucla.edu}

\author{Stefan Mitsch}
\affiliation{%
  \department{\textit{Computer Science Department}}
  \institution{Carnegie Mellon University}
  \city{Pittsburgh}
  \state{PA}
  \country{USA}
  }
\email{smitsch@cs.cmu.edu}

\author{Andr\'e Platzer}
\affiliation{%
  \department{\textit{Computer Science Department}}
  \institution{Carnegie Mellon University}
  \city{Pittsburgh}
  \state{PA}
  \country{USA}
  }
\email{aplatzer@cs.cmu.edu}


\begin{abstract}
Programmable Logic Controllers (\emph{PLCs}) provide a prominent choice of implementation platform for safety-critical industrial control systems.
Formal verification provides ways of establishing correctness guarantees, which can be quite important for such safety-critical applications.
But since PLC code does not include an analytic model of the system plant, their verification is limited to discrete properties.
In this paper, we, thus, start the other way around with hybrid programs that include continuous plant models in addition to discrete control algorithms.
Even deep correctness properties of hybrid programs can be formally verified in the theorem prover \KeYmaeraX that implements differential dynamic logic, \dL, for hybrid programs.
After verifying the hybrid program, we now present an approach for translating hybrid programs into PLC code.
The new tool, \name, implements this translation of discrete control code of verified hybrid program models to PLC controller code and, vice versa, the translation of existing PLC code into the discrete control actions for a hybrid program given an additional input of the continuous dynamics of the system to be verified. 
This approach allows for the generation of real controller code while preserving, by compilation, the correctness of a valid and verified hybrid program. PLCs are common cyber-physical interfaces for safety-critical industrial control applications, and \name serves as a pragmatic tool for bridging formal verification of complex cyber-physical systems at the algorithmic level of hybrid programs with the execution layer of concrete PLC implementations.
\end{abstract}



\maketitle

\section{Introduction}
%
%
There has been an increased emphasis on the verification and validation of software used in embedded systems in the context of \emph{industrial control systems} (ICS). 
ICS represent a class of cyber-physical systems (CPS) that provide monitoring and networked process control for safety-critical industrial environments, e.g., the electric power grid~\cite{mcgranaghan1993voltage}, railway safety ~\cite{abbpluto2016}, nuclear reactors ~\cite{kesler2011vulnerability}, and water treatment plants ~\cite{manesis1998intelligent}. 
A prominent choice of implementation platform for many ICS applications are \emph{programmable logic controllers} (PLCs) that act as interfaces between the \textit{cyber world}--i.e., the monitoring entities and process control--and the \textit{physical world}--i.e., the underlying physical system that the ICS is sensing and actuating. 
Efforts to verify the correctness of PLC applications focus on the code that is loaded onto these controllers~\cite{moon1994modeling},~\cite{darvas2015formal},~\cite{mader1999timed},~\cite{thapa2005transformation}. 
Existing methods are based on model checking of safety properties specified in modal temporal logics, e.g., Linear Temporal Logic (LTL)~\cite{gerth1995simple} and Computation Tree Logic (CTL)~\cite{clarke1986automatic}. 
However, since PLC code does not include a model of the system plant, such analyses are limited to more superficial, discrete properties of the code instead of analyzing safety properties of the resulting physical behavior. 

%
%
In this paper, we thus start from hybrid systems models of ICS, in which the discrete computations of controllers run together with the continuous evolution of the underlying physical system. 
That way, correctness properties that consider both control decisions and physical evolution can be verified in the theorem prover \KeYmaeraX~\cite{fulton2015keymaera}.
The verified hybrid programs can then be compiled to PLC code and executed as controllers.
The reverse compilation from PLC code to hybrid programs facilitates verifying existing PLC code with respect to pre-defined models of the continuous plant dynamics.

%
%
\begin{figure}[tp]
\centering
\includegraphics[width=0.48\textwidth]{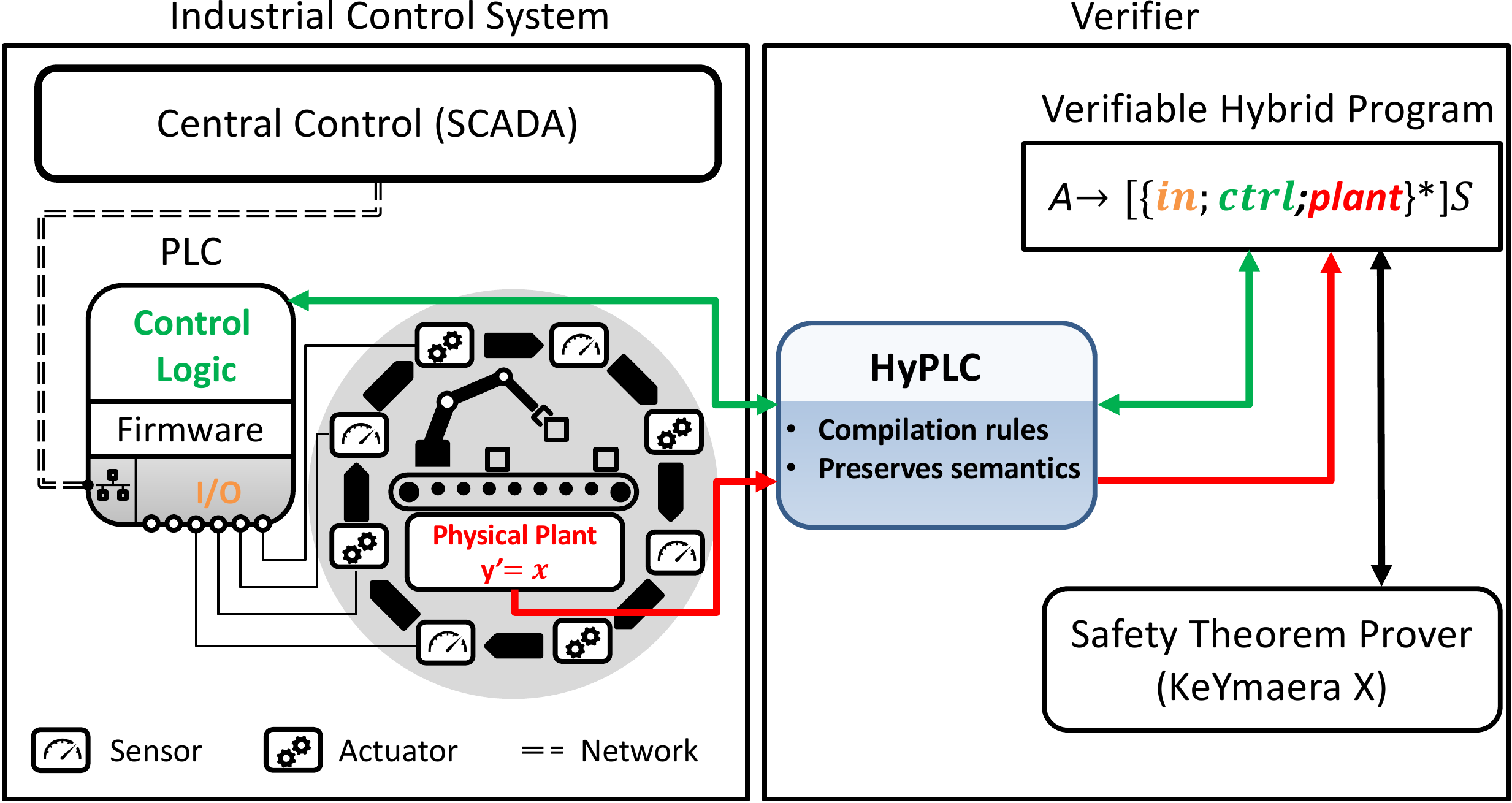}
\caption{\name provides a bidirectional translation of the discrete control of a verifiable hybrid program expressed in ~\dl{} and the control logic code that runs on a PLC in the context of a cyber-physical industrial control system
}
\label{fig:hyplc-overview}
\end{figure}
In this paper, we present \name, a tool that compiles verified hybrid systems models into PLC code and vice versa. Figure \ref{fig:hyplc-overview} depicts a high-level overview of the bidirectional compilation provided by \name.
The hybrid models are specified in differential dynamic logic, \dl\cite{DBLP:journals/jar/Platzer08,DBLP:journals/jar/Platzer17}, which is a dynamic logic for hybrid systems expressed as \textit{hybrid programs}. 
Compiling hybrid programs to PLC code generates deterministic implementations of the controller abstractions typically found in hybrid programs, which focus on capturing the safety-relevant decisions for verification purposes concisely with nondeterministic modeling concepts. 
Nondeterminism in hybrid programs can be beneficial for verification since nondeterministic models address a family of (control) programs with a single proof at once, but is detrimental to implementation with Structured Text (ST) programs on PLCs.
Therefore, in this paper we focus on hybrid programs in scan cycle form. The compilation adopts the IEC 61131-3 standards for PLCs~\cite{john2010iec}. 
Compiling PLC code to \dl and hybrid programs, implemented on top of the open-source MATIEC IEC 61131-3 compiler~\cite{sousamatiec}, provides a means of analyzing PLC code on pre-defined models of continuous evolution with the deductive verification techniques of \KeYmaeraX.
The core contributions of this paper lie in our correctness proofs for the bidirectional compilation, so that both directions of compilation yield a way of obtaining code with safety guarantees.
Finally, we evaluated our tool on a water treatment testbed~\cite{mathur2016swat} that consists of a distributed network of PLCs.



The rest of the paper is organized as follows. Section~\ref{sec:background} provides background information. Section~\ref{sec:translation-of-terms} introduces compilation rules for terms in both languages and describes how the semantics is preserved. Section~\ref{sec:translation-of-formulas} and Section~\ref{sec:translation-of-programs} describe the compilation of formulas and programs, respectively, and include formal proofs of correctness and preservation of safety across compilation. Section~\ref{sec:evaluation} presents our evaluation of \name on a water treatment case study. We discuss the limitations of \name and conclude in Section~\ref{sec:conclusion}.

\section{Preliminaries}\label{sec:background}
This section explains the preliminaries necessary to understand the underlying concepts of \name. We first provide a brief overview of PLCs, including how they are integrated into ICS as well as the associated programming languages and software model as defined by the IEC 61131-3 standard for PLCs~\cite{john2010iec}. We then discuss previous works in formal verification of PLC programs, followed by an overview of the dynamic logic and hybrid program notation used by \name.

\subsection{Programmable Logic Controllers}
Part 3 of the IEC 61131 standards~\cite{john2010iec} for PLCs specifies both the software architecture as well as the programming languages for the control programs that run on PLCs. We will provide the requisite knowledge for understanding the assumptions made by \name.

\begin{figure}[t]
\centering
\includegraphics[width=0.48\textwidth]{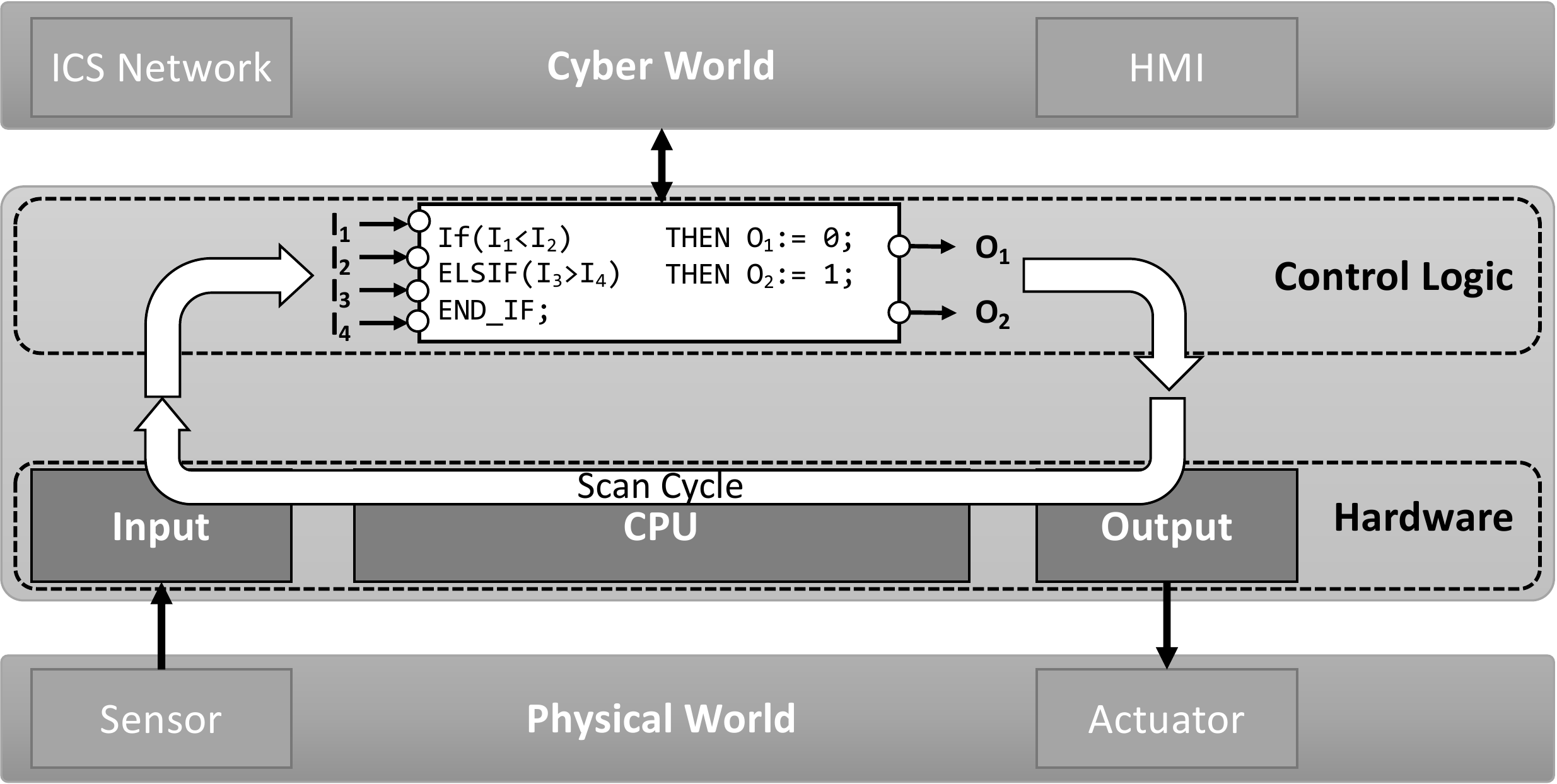}
\caption{The PLC \textit{scan cycle} in the context of ICS}
  
 \label{fig:plc-ics}
\end{figure} 

%
%

\noindent\textbf{PLCs in the context of ICS.} Figure~\ref{fig:plc-ics} shows how PLCs are integrated into ICS as well as a schematic overview of the PLC \textit{scan cycle}. Scan cycles are typical control-loop mechanisms for embedded systems. 
The PLC ``scans'' the input values coming from the physical world and processes this system state through the \textit{control logic} of the PLC, which is essentially a reprogrammable digital logic circuit. The outputs of the control logic are then forwarded through the output modules of the PLC to the physical world. 
\name focuses on hybrid programs of a shape that fits to this scan cycle control principle using time-triggered models.
\\
%
%

\noindent\textbf{Programming languages and software execution model.} \name focuses on bidirectional compilation of the \textit{Structured Text} (ST) language, which is a textual language similar to Pascal that, for formal verification purposes~\cite{darvas2017plc}, can be augmented to subsume all other languages\footnote{ladder diagrams (LD)
, function block diagrams (FBD)
, sequential function charts (SFC)
, and instruction list (IL)} defined by the IEC 61131-3 standard.  For the software execution model, we refer to~\cite{john2010iec}. We initially consider a single-resource configuration of a PLC that has a single task associated with a particular program that executes for a particular interval, ~\scduration. Because, it is a single task configuration, we initially do not consider priority scheduling.

\subsection{PLC Programming Language Verification}

Due to their wide use, there have been numerous works regarding the verification of safety properties of PLC programming languages. Rausch\etal~\cite{rausch1998formal} modeled 
PLC programs consisting only of Boolean variables, single static assignment of variables, no special functions or function blocks, and no jumps except subroutines. Such an approach was an initial attempt to provide formal verification of discrete properties of the system, i.e., properties that can be derived and verified purely from the software, ignoring the physical behavior of its plant. Similarly, other approaches have been presented whose safety properties are specified and modeled using linear temporal logic~\cite{mclaughlin2014trusted,pavlovic2007automated} or by representing the system as a finite automaton~\cite{mertke2001formal,darvas2015plcverif,tapken1998moby}.  The formal verification of such systems is limited by state-space exploration techniques, e.g., there will be an uncontrollable number of states for continuous systems because time is a variable. As such, these techniques will only be able to explore a subset of the states.
\luis{I removed instances of "shallow". However, should we also remove instances "deep"?}
\andre{Neither shallow nor deep capture what we mean, but we should make it clear that there's a crucial difference between ignoring vs including the plant. Maybe better to discuss this summarizing differences}


Conversely, there have been several works regarding the generation of PLC code based on the formal models of PLC code. PLCSpecif~\cite{darvas2015formal} is a framework for generating PLC code based on finite automata representations of the PLC. Although this framework provides a means of generating PLC code based on formally verified models, the formal verification has the aforementioned limitations of providing correctness guarantees for discrete properties of the PLC code that can be verified for a finite time horizon. The approach presented by Sacha~\cite{sacha2005automatic} has similar limitations since it uses state machines to represent finite-state models of PLC code. Darvas~\etal also used PLCSpecif for conformance checking of PLC code against temporal properties~\cite{darvas2016conformance}. Flordal\etal automatically generated PLC-code for robotic arms based on generated \textit{zone} models to ensure the arms do not collide with each other as well as to prevent deadlock situations ~\cite{flordal2007automatic}. The approach generates a finite-state model of the robot CPS environment that is then used to generate supervisory code within the PLC that controls its arm. The approach abstracts the PLC's discrete properties and does not incorporate the PLC's timing properties into the physical plant model. Furthermore, this is a domain-specific approach for robot simulation environments and does not provide generalizability nor a means of formal verification of the initially generated finite-state models. 

VeriPhy \cite{bohrer2018veriphy} compiles CPS models specified in \dl{} to verified executables that sandbox controllers with safe fallback control and monitor for expected plant behavior. 
The VeriPhy pipeline combines multiple tools to bridge implementation and arithmetic gaps and provide proofs that safety is preserved when compiling to a controller executable--while \name provides bi-directional compilation in the context of PLC scan cycles. Majumdar\etal also explored equivalence checking of C code and an associated SIMULINK model~\cite{majumdar2013compositional}. Although such an approach is useful for modelling the behavior of C code in a control system model, additional efforts are needed to interface such a model with verification tools such as \kyx as well as to model the behavior of PLCs. 
%
%
\subsection{Differential Dynamic Logic and Hybrid Programs}
\name works on models that have been specified in differential dynamic logic (\dl)~\cite{platzer2010logical,DBLP:journals/jar/Platzer17}, a logic that models hybrid systems and can be formally verified with a sound proof calculus. 
The formalized models that use \dl are referred to as \textit{hybrid programs}. As with ST, we will recall the syntax and semantics of \dl and hybrid programs as needed throughout the course of this paper.

The modal operators $\dbox{\alpha}$ and $\didia{\alpha}$ are used to formally describe the 
behavioral properties the system has to satisfy. 
If $\alpha$ denotes a hybrid program, and $\phi$ and $\psi$ are formulas, then the \dl
formula \[\phi \limply \dbox{\alpha}\psi\] means ``if $\phi$ is initially satisfied, then
$\psi$ holds true for all the states after executing the hybrid program $\alpha$''. 
This way, safety properties can be encoded for a model $\alpha$. 
We use the modeling pattern
\[
 A \limply \dbox{\prepeat{\{\text{ctrl};\text{plant}\}}}S,
\]
where $A$ represents assumptions on the initial state of the system, $\text{ctrl}$ describes the discrete control transitions of the system, $\text{plant}$ defines the continuous physical behavior of the system, and $S$ is the safety property we want to prove.
In this pattern, control and plant are repeated any number of times, as indicated with the nondeterministic repetition operator $\prepeat{}$.

We use \(\freevars{\phi}\) to refer to the free variables and \(\boundvars{\phi}\) to refer to the bound variables of formula \(\phi\) (accordingly for terms and programs) \cite{DBLP:journals/jar/Platzer17}.

\subsection{Use Case: Water Treatment Testbed}
\begin{figure}
  \centering
  \includegraphics[width=0.48\textwidth]{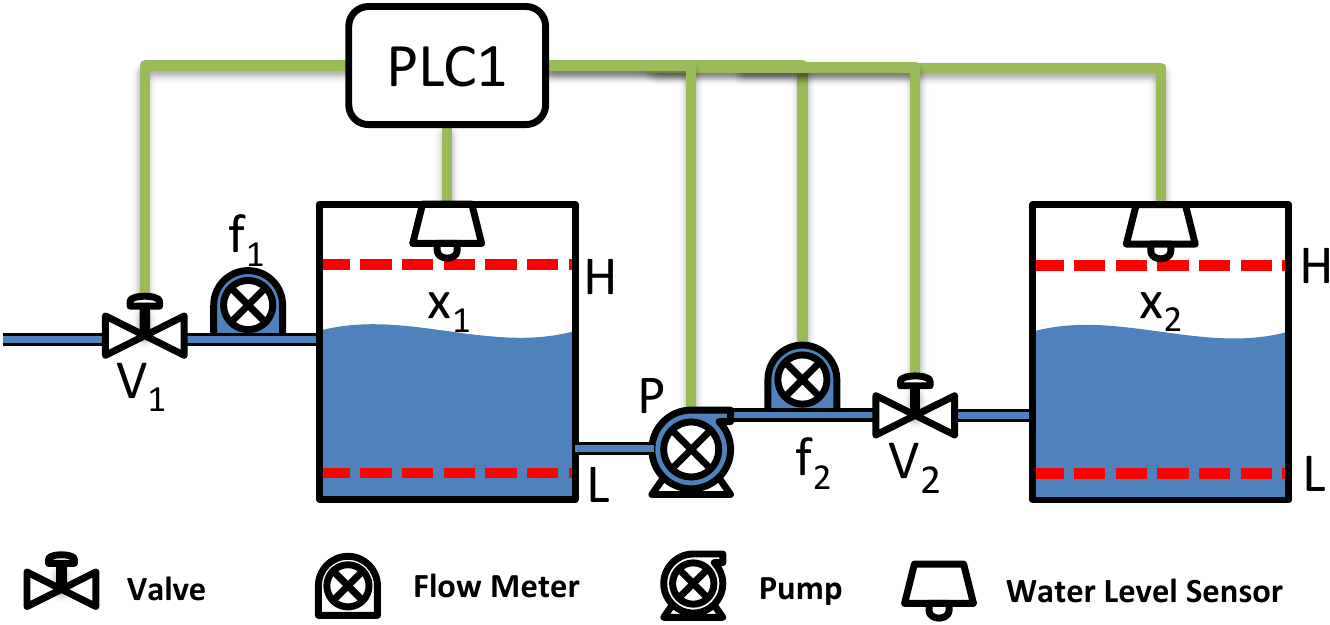}
  \caption[First Stage of Water Treatment Testbed]{The first process control components for a water treatment testbed~\cite{mathur2016swat}.}
  \label{fig:swat-simple}
\end{figure} 
As a running example for this paper, we will use a simple water tank component taken from the first of six control processes of a water treatment testbed~\cite{mathur2016swat}, depicted in Figure~\ref{fig:swat-simple}. This process is responsible for taking in water from a raw water source and feeding it into a tank. This water will then be pumped out into a second tank to be treated with chemicals. For this first process, the PLC is responsible for controlling the inflow of water for both tanks by opening or closing valves, $V_1$ and $V_2$, as well as the outflow of water to the second tank by running the pump, $P$. The PLC monitors the water level of both water tanks, $x_1$ and $x_2$, to ensure that $V_1$ and $V_2$, respectively, are closed before each respective tank overflows beyond an upper bound, $H$. Furthermore, the PLC is responsible for protecting the outflow pump, $P$, by ensuring that the pump is off if the water level of $x_1$ is below a lower threshold, $L$, or if the flow rate of the pump, $f_2$, is below a certain lower threshold, $F_L$ (not pictured in Figure~\ref{fig:swat-simple}).
\begin{figure}[t]
\centering
\begin{lstlisting}[
basicstyle=\footnotesize,
language=ST
]
PROGRAM prog0

  VAR_INPUT x1, x2, f1, f2 : REAL;  END_VAR  
  VAR_OUTPUT  V1, V2, P : BOOL;     END_VAR
  
  IF(x1 >= H1) THEN	         V1:=0;    ELSE
  IF(x1 <= L1) THEN          V1:=1;    END_IF; 
                                       END_IF;
  IF(x2 <= L2) THEN    P:=1; V2:=1;    END_IF;
  IF(x1 <= LL OR f2 <= FL OR x2 >= H2) THEN 
                       P:=0; V2:=0;    END_IF;
  
END_PROGRAM

CONFIGURATION Config0
  RESOURCE Res0 ON PLC
    TASK Main(INTERVAL:=T#1 s, PRIORITY:=0);
      PROGRAM Inst0 WITH Main : prog0;
  END_RESOURCE
END_CONFIGURATION
\end{lstlisting}
\caption[ST Program for First Process of Water Treatment Testbed]{ST program for simplified process control of the system in Figure~\ref{fig:swat-simple}}
\label{fig:swat-simple-st-code}
\end{figure}
Figure~\ref{fig:swat-simple-st-code} shows a simplified representation of the actual ST code that is loaded onto the PLC for a particular sample rate of $\scduration$ for all the associated sensors. In this model, the flow rate for the incoming raw water, $f_1$, is not incorporated into the process control. The real system simply monitors the value of this flow rate without establishing a physical dependency. The upper limits of the water tank level, $H_1$ and $H_2$, and the lower limits, $L_1$ and $L_2$, represent trigger levels that are below and above, respectively, the actual safety thresholds, $H_H$ and $L_L$. The trigger values were determined empirically in \cite{mathur2016swat}; in our proofs, we will find and verify symbolic characterizations of these trigger values. 
This simple model will be used throughout the paper to illustrate how an existing ST program can be systematically compiled to the discrete control of a hybrid program and updated if necessary to ensure the safe operation of the ICS.

\section{Compilation of Terms}\label{sec:translation-of-terms}

\noindent\textbf{Compilation approach overview.}
Compilation between ST and hybrid programs bases on two main ingredients: the syntax of the languages, given in grammars, define their notation; the language semantics give meaning to the syntactic constructs.
Compilation translates from one syntax to another, but it must be done in a way that preserves the semantics of the compiled programs.

With compilation rules, we define how to compile a term, formula, or program in the source syntax into a corresponding term, formula, or program of the target syntax.
Each rule will compile a certain program operator, and often invoke compilation on the operands.
For example, \(\compilehpeq{\phi \land \psi} \compileop \compilehpeq{\phi} ~\text{AND}~ \compilehpeq{\psi}\) compiles conjunction \(\land\) in hybrid program formulas into conjunction \(\text{AND}\) in ST of the recursively compiled sub-formulas \(\phi\) and \(\psi\).
Here, \(\compilehpeq{\phi \land \psi}\) means that we compile hybrid program formula \(\phi \land \psi\) into an ST formula; the operator \(\compileop\) describes how the compilation is done.

With proofs of compilation correctness we then show that the compilation rules preserve the semantics \emph{in a way that will allow us to conclude safety of an ST program from a safety proof of a hybrid program}.
The proofs will exploit the recursive nature of the compilation rules and apply \emph{structural induction} on the program syntax constructs, where we inductively justify each compilation rule from its easier pieces, basing on the hypothesis that the easier pieces are correctly built from the base constructs (e.g. complicated terms built from numbers and variables).

For terms and propositional formulas, the compilation rules are mostly straightforward.
The main syntactic difference is between nondeterministic choices in hybrid programs and if-then-else constructs in ST.
Aligning the semantics in the compilation correctness proofs, however, requires more work: the semantics of ST is given as an operational semantics~\cite{darvas2017plc}, which describes the effects of taking a step in a program, whereas the semantics of hybrid programs is denotational, which describes how states are related through operations of a program.

\noindent\textbf{Term compilation overview.}
In this section, we will define birectional compilation rules of the arithmetic terms in both hybrid programs and ST for PLCs.
The terms of ST are the leaf elements of ST expressions that represent the values stored in the PLC's memory and directly affect the sensing and actuation of the cyber-physical system for a particular context. As such, these values will need to be abstracted to represent the terms of an equivalent hybrid program. 
We will first discuss syntax of the terms in both languages and then define the semantics-preserving compilation.

\noindent\textbf{Notation.} We write \compilehp{\theta} for the result of compiling a hybrid program term, $\theta$, to an ST term, and we write \compilest{\theta} to represent compiling an ST term to a hybrid program term. 
We write \(\compilehpeq{\theta} \compileop s\) when $s$ is the result of compiling $\theta$ to an ST term, and \(\compilesteq{s} \compileop \theta\) when $\theta$ is the result of compiling $s$ to a \dl term.
This notation will also be used for the bidirectional compilation of formulas and programs.


\subsection{Grammar Definitions}
In order to compile terms between both languages while preserving the respective semantics, we first define the grammar for both languages.

\noindent\textbf{Grammar of ST terms.} The terms of ST considered in this paper are defined by the grammar:
\[
 \theta,\eta ::= a~\mathlarger{|}~ x~\mathlarger{|}~-\theta~\mathlarger{|}~\theta \arithmeticOp \eta ~\text{where } \arithmeticOp\ \in \{+,-,*,/\,,**\}
\]

\noindent and where $a$ is a number literal, $x \in V$ is an ST variable, and $V$ is the subset of all ST variables, and both number literals and variables are restricted to \texttt{LReal}\footnote{LReal variables are 8 byte values represented as floating points from the IEC 60559 standard.} of the numeric elementary data types defined by the IEC 61131-3 standard.

\noindent\textbf{Grammar of \dl terms.} The translatable terms of \dl and hybrid programs \cite{DBLP:journals/jar/Platzer08,DBLP:journals/jar/Platzer17} are defined by the grammar:
\[
 \theta,\eta ::= x ~\mathlarger{|}~ n
 ~\mathlarger{|}~ \theta \arithmeticOp \eta
 ~\text{where } \arithmeticOp\ \in \{+,-,\cdot,/\,,\mathbin{\char`\^}\}
\]
\noindent and where $x \in V$ is a variable and $V$ is the set of all variables. 
The grammar allows the use of number literals \(n\) as functions without arguments that are to be interpreted as the value they represent. 

Next, we provide the bidirectional compilation rules of terms and prove term compilation correctness.

\subsection{Compilation Rules}

We will first define compilation rules for the terminal expressions, referred to as atomic terms, and compose the other expressions following the recursive nature of the grammars. 

\noindent\textbf{Atomic terms.} Atomic terms in hybrid programs include variables and number literals. For the sake of simplicity, we do not consider functions within hybrid programs as we want to focus on the core elements of discrete control, and \emph{we assume that the data type \texttt{LReal} of the IEC 61131-3 standard represents mathematical reals}.
In practice, when a PLC implements \texttt{LReal} with floating point numbers, this assumption can be met with an appropriate sound encoding using, for example, interval arithmetic as verified in \cite{bohrer2018veriphy}.

\name compiles number literals and variables of hybrid programs, which evaluate to mathematical reals, to numbers and variables of data type \texttt{LReal} of the IEC 61131-3 standard as follows:
Number literals \(n\) and variables \(x\) then do not need conversion, so
\(
  \compilehpeq{n}~\compileop~n\) and \(\compilesteq{n}~\compileop~n 
\), as well as 
\(
  \compilehpeq{x}~\compileop~x\) and \(\compilest{x}~\compileop~x
\).

\noindent Next, we inductively define the compilation rules for arithmetic operations.

\noindent\textbf{Arithmetic operations.} Arithmetic operations are similarly defined in an inductive fashion in similar syntax in both languages, which makes translation of terms $\theta$ and $\eta$ straightforward as follows:
\begin{align*}
&\compilehpeq{-(x)} \compileop -(\compilehpeq{x}) &&\compilesteq{-(x)} \compileop -(\compilesteq{x})\\
&\compilehpeq{\theta \arithmeticOp \eta} \compileop \compilehpeq{\theta} \arithmeticOp \compilehpeq{\eta} &&\compilesteq{\theta \arithmeticOp \eta} \compileop \compilesteq{\theta} \arithmeticOp \compilesteq{\eta}\\
&\compilehpeq{\theta \cdot \eta} \compileop \compilehpeq{\theta} * \compilehpeq{\eta} &&\compilesteq{\theta * \eta} \compileop \compilesteq{\theta} \cdot \compilesteq{\eta}\\
&\compilehpeq{\theta \mathbin{\char`\^} \eta} \compileop \compilehpeq{\theta} \texttt{**} \compilehpeq{\eta} &&\compilesteq{\theta \texttt{**} \eta} \compileop \compilesteq{\theta} \mathbin{\char`\^} \compilesteq{\eta}
\end{align*}
\noindent where $\arithmeticOp\ \in \{+,-,/\}$.

We now provide the Lemmas for correctness of the translation of terms in both directions.
We follow \cite{darvas2017plc} and write \((\theta,\iget[state]{\I}) \stopeval c\) to express that in ST a term \(\theta\) evaluates to \(c\) in context \(\iget[state]{\I}\).
We write \(\ivaluation{\I}{\theta} = c\) to express that in \dL a term \(\theta\) evaluates to \(c\) at state \(\iget[state]{\I}\).
Details on the \dL semantics and ST semantics used in the proof can be found in Appendix~\ref{sec:dlsemantics} and Appendix~\ref{sec:structured-text-semantics}, respectively.

\begin{lemma}[Correctness of term compilation]\label{lemma:terms}
Assuming \(\texttt{LReal}=\mathbb{R}\): if \((\theta,\iget[state]{\I}) \stopeval c\) then \(\ivaluation{\I}{\compilesteq{\theta}} = c\); conversely, if \(\ivaluation{\I}{\theta} = c\) then \((\compilehpeq{\theta},\iget[state]{\I}) \stopeval c\).
\end{lemma}

\luis{TODO: need to explain why LReal represents $\mathbb{R}$}
\begin{proof}See Appendix~\ref{sec:proofs}.\end{proof}
We next define how the compilation of terms is leveraged to compile the \textit{formulas} of both languages in both directions.
\section{Compilation of Formulas}\label{sec:translation-of-formulas}

In this section, we compile modality- and quantifier-free formulas used in tests in hybrid programs and conditional expressions of ST statements. 
As was done with the terms of each language, we first discuss the syntax of the formulas for both languages.

\subsection{Grammar Definitions}

\noindent\textbf{Grammar of ST formulas.} ST formulas are used in conditional expressions defined by the IEC 61131-3 standard as follows. 
\begin{align*}
 \phi, \psi ::= \text{TRUE}&~\mathlarger{|}~\text{FALSE}~\mathlarger{|}~ \theta \relationalOpST \eta ~\mathlarger{|}~ \text{NOT}(\phi) ~\mathlarger{|}~ \phi \logicalOpST \psi\\
 \text{where } \relationalOpST\ &\in \{<,>,>=,<=,<>,=\}\\
 \text{and } \logicalOpST\ &\in \{\text{AND},\text{OR},\text{XOR}\}
\end{align*}

\noindent The values $\text{TRUE}$ and $\text{FALSE}$ represent the two Boolean values a conditional expression can take upon evaluation, $\theta$ and  $\eta$ are ST terms, operator $\relationalOpST$ ranges over relational operators used in ST, operator $\logicalOpST$ ranges over logical operators between two formulas, and \(\text{NOT}(\phi)\) is the logical negation of a formula \(\phi\). 

\noindent\textbf{Grammar of \dl{} formulas.} 
The truncated grammar for modality- and quantifier-free formulas in \dL that we consider in this paper is built using propositional connectives $\lnot$, $\land$, $\lor$, $\limply$, and $\lbisubjunct$~\cite{DBLP:journals/jar/Platzer08} as follows:
\begin{multline*}
 \phi, \psi ::= \ltrue~\mathlarger{|}~ \lfalse ~\mathlarger{|}~ \theta \relationalOpHP \eta ~\mathlarger{|}~ \neg \phi ~\mathlarger{|}~ \phi \logicalOpHP \psi\\
 \text{where } \relationalOpHP\ \in \{<,>,\geq,\leq,=\} \text{ and } \logicalOpHP\ \in \{\land,\lor,\limply,\lbisubjunct\}
\end{multline*}


\noindent and where $\theta$ and $\eta$ are \dl{} terms. 
Given these base grammars, next we present the compilation rules and the formula compilation correctness proof.

\subsection{Compilation Rules}

\noindent\textbf{Atomic formulas.} 
Atomic formulas in both languages comprise the literals \(\ltrue\) and \(\lfalse\) and comparisons of terms and are compiled in a straightforward way:
\begin{alignat*}{2}
 \compilehpeq{\ltrue} &\compileop \text{TRUE} \qquad \compilesteq{\text{TRUE}} &&\compileop \ltrue\\
 \compilehpeq{\lfalse} &\compileop \text{FALSE} \qquad \compilesteq{\text{FALSE}} &&\compileop \lfalse
\end{alignat*}

Comparisons are compiled as follows:
\begin{alignat*}{3}
 \compilehpeq{\theta=\eta} &\compileop \compilehpeq{\theta}  =  \compilehpeq{\eta} \qquad &&\compilesteq{\theta=\eta} \compileop \compilesteq{\theta}  =  \compilesteq{\eta}\\
 \compilehpeq{\theta\neq\eta} &\compileop \compilehpeq{\theta}  <>  \compilehpeq{\eta} &&\compilesteq{\theta <> \eta} \compileop \compilesteq{\theta}  \neq  \compilesteq{\eta}\\
 \compilehpeq{\theta>\eta} &\compileop \compilehpeq{\theta}  >  \compilehpeq{\eta} &&\compilesteq{\theta>\eta} \compileop \compilesteq{\theta}  >  \compilesteq{\eta}\\
 \compilehpeq{\theta\geq\eta} &\compileop \compilehpeq{\theta} >= \compilehpeq{\eta} &&\compilesteq{\theta >= \eta} \compileop \compilesteq{\theta} \geq \compilesteq{\eta}\\
 \compilehpeq{\theta<\eta} &\compileop \compilehpeq{\theta} < \compilehpeq{\eta} &&\compilesteq{\theta<\eta} \compileop\compilesteq{\theta} < \compilesteq{\eta}\\
 \compilehpeq{\theta\leq\eta} &\compileop \compilehpeq{\theta} <= \compilehpeq{\eta} &&\compilesteq{\theta <= \eta} \compileop \compilesteq{\theta} \leq \compilesteq{\eta}
\end{alignat*}

The compilation rules for the atomic formulas are the basis for compiling compositional formulas.

\noindent\textbf{Logical formulas.} Logical connectives \(\lnot,\land,\lor\) are straightforward, whereas \(\limply,\lbisubjunct\) are rewritten in terms of \(\land,\lor\) before compilation (similar for \text{XOR}):
\[
\compilehpeq{\lnot(\phi)} \compileop \text{NOT}(\compilehpeq{\phi}) \qquad \compilesteq{\text{NOT}(\phi)} \compileop \lnot(\compilesteq{\phi})
\]
\begin{alignat*}{3}
 &\compilehpeq{\phi \land \psi} &&\compileop\compilehpeq{\phi} ~\text{AND}~ \compilehpeq{\psi}\\
 &\compilesteq{\phi ~\text{AND}~ \psi} &&\compileop \compilesteq{\phi} \land \compilesteq{\psi}\\
 &\compilehpeq{\phi \lor \psi} &&\compileop \compilehpeq{\phi}  ~\text{OR}~ \compilehpeq{\psi}\\
 &\compilesteq{\phi ~\text{OR}~ \psi} &&\compileop \compilesteq{\phi} \lor  \compilesteq{\psi}\\
& \compilehpeq{\phi \limply \psi} &&\compileop\compilehpeq{\lnot\phi \lor \psi}\\
& \compilehpeq{\phi \lbisubjunct \psi} &&\compileop \compilehpeq{(\lnot \phi \land \lnot \psi) \lor (\phi \land \psi)}\\
&\compilesteq{\phi ~\text{XOR}~ \psi} && \compileop \compilesteq{\phantom{\text{OR}~}(\text{NOT}(\phi) ~\text{AND}~ \psi)\\ 
&&&\phantom{\compilesteq{(~}} \text{OR}~ (\text{NOT}(\psi) ~\text{AND}~ \phi)}
\end{alignat*}
\luis{TODO: fix "awkward" order; maybe put all ST first and then HP, but then it would take up a significant amount of space}

We now prove correctness of the compilation of formulas in both directions.
In ST, we write \((\phi,\iget[state]{\I}) \stopeval \top\) and in \dL \(\iget[state]{\I} \models \phi\) to say that formula \(\phi\) is true at state \(\iget[state]{\I}\).
\begin{lemma}[Correctness of formula compilation]\label{lemma:formula}
Formulas evaluate equivalently: \(\iget[state]{\I} \models \phi\) iff \((\compilehpeq{\phi},\iget[state]{\I}) \stopeval \top\) and, conversely, \((\phi,\iget[state]{\I}) \stopeval \top\) iff \(\iget[state]{\I} \models \compilesteq{\phi}\).
\end{lemma}
\begin{proof}See Appendix~\ref{sec:proofs}.\end{proof}

\section{Compilation of Programs}\label{sec:translation-of-programs}

Now that we know how to correctly compile terms and formulas in both languages, we turn to compiling program constructs.
Since these programs, when executed on a PLC, interact with the physical world, our overall goal is to provably establish safety properties of the physical behavior of an ICS.
To this end, we again show compilation correctness with respect to the semantics of the languages, which will serve as a stepping stone to describe the program effect in the larger context of the PLC scan cycle.

We first provide an overview of our hybrid system model of a PLC scan cycle, before we introduce the grammars and compilation rules for both languages and prove compilation correctness.

\subsection{Scan Cycle Hybrid System Model}
\begin{figure}[htb]
\centering
\includegraphics[width=\columnwidth]{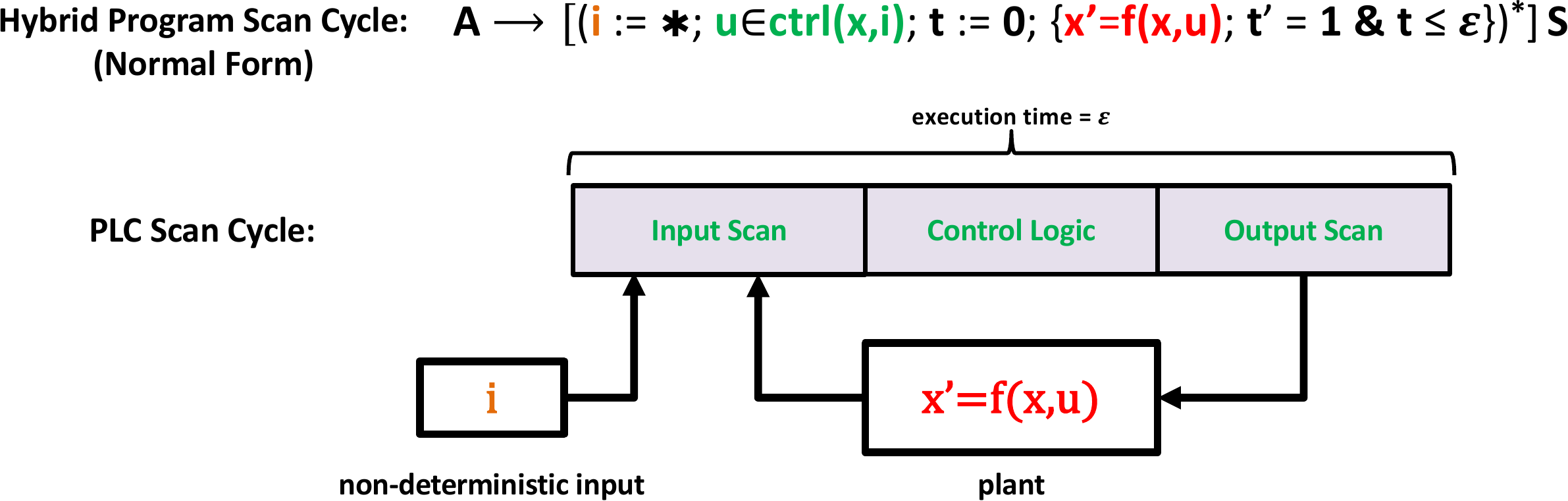}
\caption{Hybrid system model of a PLC scan cycle}
\label{fig:scan-cycle-translation}
\end{figure}
We model the PLC scan cycle as a hybrid program of a particular shape---referred to as a hybrid program in \emph{scan cycle normal form}---in order for safety properties verified about a hybrid program to directly transfer to its implementation in ST. 

Figure~\ref{fig:scan-cycle-translation} provides an overview of the components of a hybrid program in scan cycle normal form and how they relate to a PLC scan cycle. 
A PLC scan cycle is a periodic process that, on each iteration, scans the inputs, then executes the control logic to set outputs, and finally forwards outputs to the actuators. 
The total scan cycle duration in this abstracted model is $\scduration$.

Our hybrid program model of such a scan cycle uses nondeterministic assignments $\extinput$ to model arbitrary external input to the PLC system, such as sensor values whose state cannot be estimated or user input from a user interface.
Based on the current state \(x\) and inputs \(i\), the controller \(\ctrl\) then chooses control actions \(u\) from a set of possible choices.
The plant modeled by $\plant$ continuously evolves the system state \(x\) according to the control actions along the differential equations \(\D{x}=f(x,u)\) and keeps track of the scan cycle duration bound \(\scduration\) with a clock \(t\).

\begin{definition}[Scan cycle normal form]\label{def:scancyclenormalform}
We call a hybrid program with shape \(\hpscancycle\) a program in \emph{scan cycle normal form}.
It is safe, if formula
\(A \limply \dbox{\prepeat{(\hpscancycle)}}S\) is valid.
\end{definition}

In the following subsections, we detail how a controller \(\ctrl\) is translated into an ST program and its associated configuration.
We will leave the generation of code for nondeterministic inputs and physical plant components (e.g., monitors that check model and true system execution for compliance) as future work. 

We will use the operational semantics of ST and dynamic semantics of hybrid programs to ensure that the compilation preserves meaning.
Additionally, we will use the static semantics of hybrid programs in terms of their bound and free variables to derive configuration information for the PLC code (e.g., distinguish between input and output variables).


\subsection{Grammar Definitions}
We present the respective grammars for programs in each language.

\noindent\textbf{Grammar of ST programs.} ST programs refer to the sequence of statements defined by the IEC 61131-3 standard that form entire ST programs. 
We consider ST statements \stprogram{1} and \stprogram{2} as follows:
\begin{align*}
\stprogrameq{1}, \stprogrameq{2} ::= x := \theta ~\mathlarger{|}~ &\stif{\phi}{\stprogrameq{1}}{\stprogrameq{2}} ~\mathlarger{|}\\
& \stifthen{\phi}{\stprogrameq{1}}
~\mathlarger{|}~ \stprogrameq{1};\stprogrameq{2}
\end{align*}

Where $x := \theta$ is assignment of an ST term $\theta$ to variable $x$, $\stif{\phi}{\stprogrameq{1}}{\stprogrameq{2}}$ is a conditional statement where $\stprogrameq{1}$ is executed if $\phi$ is true and $\stprogrameq{2}$ is executed otherwise, and $\stprogrameq{1};\stprogrameq{2}$ is the sequential composition of ST programs where $\stprogrameq{2}$ executes after $\stprogrameq{1}$ has finished its execution. While the structured grammar can support several other control structures such as finitely bounded loops and case-statements, these structures can be represented as a series of if-then-else statements. The \dl{} grammar is composed in a similar fashion. 

\noindent\textbf{Grammar of \dl{} programs.} The grammar for PLC-translatable \dl{} hybrid programs 
is defined as follows.
\begin{align*}
 \alpha, \beta~~::=~~\humod{x}{\theta}~~\mathlarger{|}~~ \pchoice{(\ptest{\phi};\alpha)}{\beta}~~\mathlarger{|}~~
& \pchoice{(\ptest{\phi};\alpha)}{(\ptest{\lnot\phi};\beta)} ~~\mathlarger{|}\\
& \pchoice{(\ptest{\phi};\alpha)}{\ptest{\lnot\phi}}~~\mathlarger{|}~~ \alpha;\beta
\end{align*}
\noindent Where $\humod{x}{\theta}$ are assignments of the value of a term $\theta$ to the variable $x$, $\pchoice{(\ptest{\phi};\alpha)}{\beta}$ is a guarded execution of \(\alpha\) (possible if \(\phi\) is true) and default \(\beta\) (can be executed nondeterministically regardless of \(\phi\) being true or false), \(\pchoice{(\ptest{\phi};\alpha)}{(\ptest{\lnot\phi};\beta)}\) is an if-then-else conditional statement, \(\pchoice{(\ptest{\phi};\alpha)}{\ptest{\lnot\phi}}\) is an if-then conditional statement without else, and \(\alpha;\beta\) is a sequential composition
~\cite{DBLP:journals/jar/Platzer08,DBLP:journals/jar/Platzer17}. 
Given these base grammars for the programs, we now present the compilation rules and the associated correctness proofs that will allow us to conclude safety of ST programs from safety proofs of hybrid programs. 
Preserving safety will allow us to compile existing ST programs into hybrid programs and analyze their interaction with the physical plant for safety, and conversely compile the controllers of hybrid programs into ST programs for execution on a PLC.

\subsection{Compilation Rules}

\noindent\textbf{Deterministic assignment.} 
Assignments of terms to variables in hybrid programs represent the core of discrete state transitions in a hybrid system. 


The syntax and operational effect of a discrete assignment is the same in both languages, so compilation is straightforward:
\begin{align*}
 \compilehpeq{\humod{x}{\theta}} &\compileop \compilehpeq{x} := \compilehpeq{\theta}\\
 \compilesteq{x := \theta} &\compileop \humod{\compilesteq{x}}{\compilesteq{\theta}}
\end{align*}

The static semantics of discrete assignments in hybrid programs provides information about input and output variables of the generated ST code: an assignment contributes \(\boundvars{\humod{x}{\theta}} = \{x\}\) to the set of output variables, and \(\freevars{\humod{x}{\theta}}=\freevars{\theta}\) to the set of input variables \cite{DBLP:journals/jar/Platzer17}.

\noindent\textbf{Sequential composition programs.} The sequential composition of two hybrid  programs $\alpha$ and $\beta$ lets the hybrid program $\beta$ starts executing after $\alpha$ has finished, meaning that $\beta$ never starts if the program $\alpha$ does not terminate. 
Sequential composition of ST statements has identical meaning, and so compilation between ST and hybrid programs is straightforward as follows:
\[
 \compilehpeq{\alpha;\beta} \compileop \compilehpeq{\alpha}  ; \compilehpeq{\beta} \qquad \compilesteq{\alpha;\beta} \compileop \compilesteq{\alpha}  ; \compilesteq{\beta}
\]

A sequential composition contributes the input and output variables of both its sub-programs: it has output variables \(\boundvars{\alpha;\beta}=\boundvars{\alpha} \cup \boundvars{\beta}\) and input variables \(\freevars{\alpha;\beta}=\freevars{\alpha} \cup (\freevars{\beta} \setminus \mustboundvars{\alpha})\). 
Note that the input variables are not simply the union of both sub-programs, since some of the free variables of \(\beta\) might be bound on all paths in \(\alpha\)---in \(\mustboundvars{\alpha}\)---and therefore no longer be free in the sequential composition \cite{DBLP:journals/jar/Platzer17}.

\begin{remark}[ST Task Execution Timing]
 The execution of a series of statements with respect to sequential composition assumes that the statements execute atomically, which is defined in the transition semantics of hybrid programs. We do not model the preemption of higher priority tasks as the modeling of the PLC's task scheduling is beyond of the scope of this paper and left for future research. 

\name assumes that the developer designs a system with multiple tasks such that (1) the execution time of a highest priority task is  less than its period and that (2) the total execution of all tasks is less than the period of the lowest priority tasks~\cite{allenbradley2008}. 
\end{remark}

\noindent\textbf{Conditional programs.} 
In the translatable fragment of hybrid programs we allow tests to occur only as the first statement of the branches in nondeterministic choices, and we allow only nondeterministic choices that are guarded with tests.
A nondeterministic choice between hybrid programs $\ptest{\phi};\alpha$ and $\beta$ executes either hybrid program and is resolved on a PLC by favoring execution of $\ptest{\phi};\alpha$ over $\beta$ in an if-then-else statement. 
The test statement \(\ptest{\phi}\) in the beginning avoids backtracking. 
The compilation is defined as follows.
\begin{align*}
\compilehpeq{\pchoice{(\ptest{\phi};\alpha)}{\beta}} & \compileop \stif{\compilehpeq{\phi}}{\compilehpeq{\alpha}}{\compilehpeq{\beta}}\\ 
\compilehpeq{\pchoice{(\ptest{\phi};\alpha)}{(\ptest{\lnot\phi};\beta)}} & \compileop \stif{\compilehpeq{\phi}}{\compilehpeq{\alpha}}{\compilehpeq{\beta}}\\ 
\compilehpeq{\pchoice{(\ptest{\phi};\alpha)}{\ptest{\lnot\phi}}} & \compileop \stifthen{\compilehpeq{\phi}}{\compilehpeq{\alpha}}
\end{align*}

The static semantics combines the input and output variables of both programs: output variables \(\boundvars{\pchoice{(\ptest{\phi};\alpha)}{\beta}} = \boundvars{\pchoice{(\ptest{\phi};\alpha)}{(\ptest{\lnot\phi};\beta)}} = \boundvars{\alpha} \cup \boundvars{\beta}\) and input variables \(\freevars{\pchoice{(\ptest{\phi};\alpha)}{\beta}} = \freevars{\pchoice{(\ptest{\phi};\alpha)}{(\ptest{\lnot\phi};\beta)}} = \freevars{\phi} \cup \freevars{\alpha} \cup \freevars{\beta}\).

Because we only consider loop-free semantics, we avoid having to enforce backtracking for deep tests that may exist in $\alpha$ or $\beta$. Instead, the tests will simply be compiled as nested conditional programs. 
\andre{''backtracking for deep tests..." unclear}

ST conditional programs statements compile to guarded nondeterministic choices in hybrid programs as follows:
\begin{align*}
&\compilesteq{\stif{\phi}{\alpha}{\beta}}~\compileop\\ 
& \qquad \pchoice{(\ptest{\compilesteq{\phi}};\compilesteq{\alpha})}{(\ptest{\lnot\compilesteq{\phi}};\compilesteq{\beta})}\\
&\compilesteq{\stifthen{\phi}{\alpha}} \compileop
\pchoice{(\ptest{\compilesteq{\phi}};\compilesteq{\alpha})}{\ptest{\lnot\compilesteq{\phi}}}
\end{align*}

Next, we prove compilation correctness that will allow us to transfer safety proofs of hybrid programs to ST programs.
We write \((\stprg{\stprogrameq{1}},\iget[state]{\I}) \stprgeval (\stprg{\stprogrameq{2}},\iget[state]{\It})\) to say that program \(\stprogrameq{1}\) executed in context \(\iget[state]{\I}\) transitions to a new context \(\iget[state]{\It}\) with remaining program \(\stprogrameq{2}\).
We write \((\iget[state]{\I},\iget[state]{\It}) \in \iaccess[\alpha]{\I}\) to say that the final state \(\iget[state]{\It}\) is reachable from the initial state \(\iget[state]{\I}\) by running the hybrid program \(\alpha\).

\begin{lemma}[Correctness of ST to HP compilation]\label{lemma:st-to-hp}
All states reachable with the ST control program are also reachable by the target hybrid program: If \((\stprg{\stprogrameq{1}},\iget[state]{\I}) \stprgeval (\stprg{\stskip},\iget[state]{\It})\) then \((\iget[state]{\I},\iget[state]{\It}) \in \iaccess[\compilesteq{\stprogrameq{1}}]{\I}\) for all \(\iget[state]{\I},\iget[state]{\It}\), where \textbf{skip} denotes the end of code for a scan cycle.
\end{lemma}
\begin{proof}See Appendix~\ref{sec:proofs}.\end{proof}

\begin{lemma}[Correctness of HP to ST compilation]\label{lemma:hp-to-st}
All states reachable with the ST control program are also reachable by the source hybrid program:
If \((\stprg{\compilehpeq{\alpha}},\iget[state]{\I}) \stprgeval (\stprg{\stskip},\iget[state]{\It})\) then \((\iget[state]{\I},\iget[state]{\It}) \in \iaccess[\alpha]{\I}\) for all \(\iget[state]{\I},\iget[state]{\It}\).

\end{lemma}
\begin{proof}See Appendix~\ref{sec:proofs}.\end{proof}

\subsection{Preserving Safety Guarantees across Compilation}

Correct compilation guarantees that safety properties verified for hybrid programs in scan cycle normal form shape are preserved for the runs of translated ST programs.
Def.~\ref{def:stprogramrun} expresses how a loop-free ST program is executed repeatedly in the scan cycle of a PLC, connected to inputs, and drives the plant through its results.

\begin{definition}[Run of ST program]\label{def:stprogramrun}
A sequence of states\\ \(\iget[state]{\Ist}_0,\iget[state]{\Ist}_1,\iget[state]{\Ist}_2,\ldots,\iget[state]{\Ist}_n\) is a run of ST program \(\stprogrameq{1}\) with input (variable vector) $i$ and plant \(\plant\) with scan cycle duration \(\scduration\) iff for all \(i{<}n\) the program executes to completion \((\stprg{\stprogrameq{1}},\iget[state]{\Isig}_i) \stprgeval (\stprg{\stskip},\iget[state]{\I}_i)\) for some program start state \(\iget[state]{\Isig}_i\) obtained from the previous state \(\iget[state]{\Ist}_i\) in the run by reading input s.t. \((\iget[state]{\Ist}_i,\iget[state]{\Isig}_i) \in \iaccess[\extinput]{\I}\) and some program result state \(\iget[state]{\I}_i\) driving the plant to the next state \(\iget[state]{\Ist}_{i+1}\) in a continuous transition of duration $t{\leq}\scduration$ s.t. \((\iget[state]{\I}_i,\iget[state]{\Ist}_{i+1}) \in \iaccess[\plant]{\I}\).
\end{definition}

Def.~\ref{def:stprogramrun} expresses how a ST program interacts with the physical world; Def.\ref{def:scancyclenormalform} says that a hybrid program in scan cycle normal form, \(\hpscancycle\), is safe if it reaches only safe states in which \(S\) is true when started in states where \(A\) is true.
We now translate safety to the compiled ST program.
Intuitively, a hybrid program is compiled safe to ST when any ST program run that starts in a state matching the assumptions \(A\) reaches only states where running the plant is safe \(S\), as expressed in Theorem~\ref{theorem:safety}.

\begin{theorem}[Compilation Safety]\label{theorem:safety}
If the \dl formula
\[A \limply \dbox{\prepeat{(\hpscancycle)}}S\] 
is valid, and a run \(\iget[state]{\Ist}_0,\iget[state]{\Ist}_1,\iget[state]{\Ist}_2,\ldots,\iget[state]{\Ist}_n\) of \(\compilehpeq{\ctrl}\) with input $i$ and plant \(\plant\) starts with satisfied assumptions \(\iget[state]{\Ist}_0 \models A\), then \(\iget[state]{\Ist}_i \models S\) for all \(i\).
\end{theorem}
\begin{proof}
By Lemma \ref{lemma:hp-to-st}: if \((\stprg{\compilehpeq{\ctrl}},\iget[state]{\Isig}_i) \stprgeval (\stprg{\stskip},\iget[state]{\I})\) then \((\iget[state]{\Isig}_i,\iget[state]{\I}) \in \iaccess[\ctrl]{\I}\).
Since \(\iget[state]{\Ist}_0,\ldots,\iget[state]{\Ist}_n\) is a run of \(\compilehpeq{\ctrl}\) in input $i$ and plant \(\plant\), we have for all \(i{<}n\) that \((\iget[state]{\Ist}_i,\iget[state]{\Isig}_i) \in \iaccess[\extinput]{\Ist}\), \((\iget[state]{\Isig}_i,\iget[state]{\I}_i) \in \iaccess[\ctrl]{\Isig}\), and \[(\iget[state]{\I},\iget[state]{\Ist}_{i+1}) \in \iaccess[\plant]{\I}\enspace .\]
Thus, by the semantics of sequential composition \cite{DBLP:journals/jar/Platzer17}, 
\begin{multline*}
(\iget[state]{\Ist}_i,\iget[state]{\Ist}_{i+1}) \in\\ \iaccess[\hpscancycle]{\Ist}
\end{multline*}
for all \(i{<}n\).
Hence, we conclude \(\iget[state]{\Ist}_i \models S\) for all \(i\) by the validity of \[A \limply \dbox{\prepeat{(\hpscancycle)}}S
\qedhere\]
\end{proof}

Theorem~\ref{theorem:safety} means that an ST program enjoys the safety proof of a hybrid program if our compilation was used in the process (either the hybrid program used in the proof was compiled from the ST program, or the hybrid program was the source for compiling the ST program).
Next, we analyze the shape and static semantics of a hybrid program in scan-cycle normal form to extract configuration information.

\subsection{Cyclic Control Configuration}

ST programs are complemented with a configuration that structures the programs into tasks, assigns priorities and execution intervals to these tasks, and allocates computation resources for the tasks.
For a hybrid program in scan-cycle normal form per Def.~\ref{def:scancyclenormalform}
\[
\prepeat{(\hpscancycle)}
\]
do its shape and static semantics provide essential insight into the required configuration information.
The modeling pattern in the scan-cycle normal form is that of a time-triggered repetition, achieved by a clock variable $t$ that is reset to \(0\) before the continuous dynamics, evolves with constant slope \(1\), and allows following the continuous dynamics for up to $\scduration$ time. 
The combined effect is that the input \(\extinput\) and control \(\ctrl\) are executed at least once every \(\scduration\) time. 
In the compilation setup, a value for \(\scduration\) must be provided (e.g., with a formula \(\scduration=n\) as part of the assumptions \(A\) in the safety proof) and is taken as the scan cycle configuration of a PLC.

For a single task, we define the compilation of a safety property of a hybrid program in scan-cycle normal form to a task as:
\begin{multline*}
\compilehpeq{A \limply \dbox{\prepeat{(\hpscancycleabbrv)}}S}
~\compileop \\ \texttt{Task}(\compilehpeq{\ctrl},\scduration),\\
 \text{where}~\plantabbrv \equiv \plant
\end{multline*}
\noindent and \texttt{Task}($\alpha$,$\epsilon$) is a shorthand defining a task\footnote{A task is being used here to abstract the other configuration components of an ST program, i.e., Configurations and Resources. We assume only one configuration and one resource at a time in this paper for a single PLC.} that executes $\alpha$ (here the discrete control $\ctrl$ translated to ST), cyclically with an interval $\scduration$. 
Similarly, we define the converse compilation of a task with an ST program $\alpha$--whose variables $i$ are of type \lstinline{VAR_INPUT} from the configuration--and execution time of $\scduration$ as
\begin{multline*}
\compilesteq{\texttt{Task}(\alpha,\scduration)}~\compileop \\ A \limply \dbox{\prepeat{(\prandom{i};\compilesteq{\alpha};\plantabbrv)}}S\\
 given~\plantabbrv \equiv \plant
\end{multline*}
Since the ST program does not include an analytic plant model, the compiled controller is augmented with the differential equations from a plant given as extra input.
The sets of input and output variables determined by analyzing the static semantics of the hybrid program inform the program configuration variable declaration blocks \lstinline{VAR_INPUT} and \lstinline{VAR_OUTPUT}, as seen in Figure~\ref{fig:swat-simple-st-code}.

\noindent\textbf{Extension to multiple tasks.} A future extension to multiple tasks would consider a single configuration of a PLC with a single resource that has a one-to-one mapping of task configurations to ST programs. A designated clock $t_n$ per task would keep track of the associated task's execution interval $\scduration_n$. 
The task execution intervals would be checked periodically every $\scduration_\text{sc}$ times, which represents the scan cycle timing of the PLC. 
Any task with elapsed clock $t_n \geq \scduration_n\) is executed (which means that tasks are executed with at most $\scduration_\text{sc}$ delay).


\section{Evaluation}\label{sec:evaluation}

Now that we have provided the compilation rules that will be used by \name, we evaluate the tool on a real system. \name was implemented as two module extensions for the \kyx{} tool: one for each compilation direction. For the compilation of hybrid programs to ST, the compilation rules were implemented on top of the existing parser of the \kyx{} tool. Given the abstract syntax tree of a hybrid program, \name generates the associated ST code based on the compilation rules. 
The implementation was written in Scala with $\sim$700 LoC. 

Similarly, the module for the compilation of an ST program to a hybrid program was implemented on top of the lexical analysis provided by the MATIEC IEC 61131-3 compiler~\cite{sousamatiec}. The MATIEC compiler provides modules that compile ST programs to either C code or other languages provided by the IEC standard. With the same APIs we implemented the compilation rules from the abstract syntax tree of an ST program. The module was implemented in \texttt{C++} with $\sim$1000 LoC. 

We next present how \name was evaluated against the water treatment testbed.
\subsection{Use Case: Water Treatment Testbed}
In the case study, we first compiled the PLC code from the water treatment testbed shown in Figure~\ref{fig:swat-simple-st-code} into a hybrid program. 
Formal verification in \KeYmaeraX showed that this implementation is unsafe. 
We then updated the generated hybrid program with the necessary assumptions to guarantee the safety of the ICS. Finally, we compiled the fixed hybrid program into PLC code that, by Theorem~\ref{theorem:safety}, enjoys the safety proof of the hybrid program.

\subsubsection{Counterexamples in Existing PLC Code}
In order to compile the ST controller into a hybrid program of the water treatment testbed, we provide the continuous plant of the ICS in terms of differential equations, as well as the initial state constraints $A$. These are combined with the compiled $\text{ctrl}$ of the ICS that provides the discrete-state transitions of the system. Finally, we define the safety requirement, $S$, that ensures that the water tank levels always remain within their upper ($H_H$) and lower ($L_L$) thresholds.

\begin{figure}[th]
\centering
\small
\boxalign{\begin{align*}
A &\limply [\prepeat{\{\text{in};\text{ctrl};\humod{t}{0};\{\text{plant}~\&~\ivr\}\}}]S\\
A &\equiv  L_1\leq x_1 \land x_1\leq H_1 \land L_2 \leq x_2 \land x_2 \leq H_2\\
 &\phantom{\equiv~} \land V_1=0 \land V_2=0 \land P=0\\
 &\phantom{\equiv~} \land \scduration\geq 0 \land F_L > 0 \land L_L < L_1 \land L_L < L_2 \\
 &\phantom{\equiv~} \land L_1 < H_1 \land L_2 < H_2 \land H_1 < H_H \land H_2 < H_H\\
\text{in} & \equiv \prandom{f_1};~ \prandom{f_2}\\
\text{ctrl} &\equiv\, \{\phantom{\cup\,} \ptest{\phantom{\lnot}(x_1 \geq H_1)};~ \humod{V_1}{0}\\ 
     &\phantom{\equiv~\{}\cup\, \ptest{\lnot(x_1 \geq H_1)};~ \{ \phantom{\cup\,} \ptest{\phantom{\lnot}(x_1 \leq L_1)}; \humod{V_1}{1}\\
     &\phantom{\equiv~\{\cup\, \ptest{\lnot(x_1 \geq H_1)};~ \{} \cup\, \ptest{\lnot(x_1 \leq L_1 )}\}\\
     &\phantom{\equiv~} \};\\
     &\phantom{\equiv~} \{\phantom{\cup\,}\ptest{\phantom{\lnot}(x_2 \leq L_2)};~ \humod{P}{1};~ \humod{V_2}{1}\\
     &\phantom{\equiv~\{}\cup\, \ptest{\lnot(x_1 \leq L_2)}\};\\
     &\phantom{\equiv~} \{ \phantom{\cup\,} \ptest{(x_1 < L_L \lor f_2 \leq F_L \lor x_2 > H_2)};~ \humod{P}{0};~ \humod{V_2}{0}\\
     &\phantom{\equiv~\{} \cup\, \ptest{\lnot(x_1 < L_L \lor f_2 \leq F_L \lor x_2 > H_2)}\}\\
\text{plant} &\equiv \D{x}_1=V_1 \cdot f_1 - V_2 \cdot P \cdot f_2\syssep~ \D{x}_2=V_2 \cdot P \cdot f_2 \syssep \D{t}=1 \\
\ivr &\equiv t \leq \scduration \land x_1 \geq 0 \land x_2 \geq 0 \land f_1\geq0 \land f_2 \geq 0\\
S &\equiv L_L \leq x_1 \land x_1 \leq H_H \land L_L \leq x_2 \land x_2 \leq H_H
\end{align*}}
\caption[Hybrid Program Translation of Original Water Treatment Control]{Hybrid program generated by \name. This is a compilation of the PLC code in Figure~\ref{fig:swat-simple-st-code}}
\label{fig:generated-hp-program}
\end{figure}

Figure~\ref{fig:generated-hp-program} shows the full hybrid program generated by \name that incorporates both the compiled ST code as well as the continuous dynamics of the water treatment testbed. 
Intuitively, this model cannot be proven as there are no constraints on the flow rates $f_1$ and $f_2$, nor do the guards on actuation enforce such constraints. 
We use \KeYmaeraX and the \dL proof calculus to find counterexamples for the faulty combinations of operating the valves \(V_1\) and \(V_2\), both for concrete threshold values \cite{gohAdepuJunejoMathur} and the generalized threshold conditions \(L_L < L_1 < H_1 < H_H \land L_L < L_2 < H_2 < H_H\) of Figure~\ref{fig:generated-hp-program}. 
Some representative examples are listed below:
\begin{itemize}
\item If \(x_1 \geq H_1\) (so \(V_1=0\)) and \(x_2 \leq H_2\) (so \(V_2=1\)): without time and flow rate bounds, the pump may drain the first tank when it attempts to protect underflow in the second tank; it may also cause overflow of the second tank.
\item If only \(V_1=1\) is open, the first tank may overflow.
\item If both valves are open, either tank may overflow, or the first tank may underflow, depending on the ratio of flow rates.
\end{itemize}
\KeYmaeraX finds such counterexamples by unrolling the loop and analyzing paths through the loop body to (i) collect assumptions (e.g., conditions in tests \(x_1 \geq H_1\), and effects of assignments \(V_1=1\) from \(\humod{V_1}{1}\)) and (ii) propagate program effects into proof obligations (e.g., the effect of the flow rate and valves on the water level \(\D{x}_1=V_1 \cdot f_1-V_2 \cdot P \cdot f_2\) is propagated into \(S\)).
A counterexample consists of sample values for the variables such that the collected assumptions are satisfied but the proof obligations are not. 
Analyzing these sample values point to potential fixes (e.g., no flow into the first tank \(f_1=0\) with simultaneous large out flow \(f_2\) indicates that the valve \(V_2\) must be turned off before the first tank drains entirely).

\begin{figure*}
\centering
\boxalign{\begin{align*}
A &\limply [\prepeat{\{\text{in};\pmb{\text{ctrl}};\humod{t}{0};\{\text{plant}~\&~\ivr\}\}}]S\\
\text{ctrl} &\equiv\, \{\phantom{\cup\,} \ptest{\phantom{\lnot}(\pmb{f_1 > (H_H-x_1)/\scduration})};~ \humod{V_1}{0}\\ 
     &\phantom{\equiv~\{}\cup\, \ptest{\lnot(\pmb{f_1 > (H_H-x_1)/\scduration})};~ \{ \ptest{(x_1 \leq L_1)}; \humod{V_1}{1} ~\cup~ \ptest{\lnot(x_1 \leq L_1 )}\}~\}\\
     &\phantom{\equiv~} \{~\ptest{(x_2 \leq L_2)};~ \humod{P}{1};~ \humod{V_2}{1} ~\cup~ \ptest{\lnot(x_1 \leq L_2)}~\}\\
     &\phantom{\equiv~} \{ \phantom{\cup\,} \ptest{\phantom{\lnot}(\pmb{V_1 \cdot f_1-V_2 \cdot P \cdot f_2 < (L_L-x_1)/\scduration} \lor f_2 \leq FL \lor \pmb{V_2 \cdot P \cdot f_2 > (H_H-x_2)/\scduration})};~ \humod{P}{0};~ \humod{V_2}{0}\\
     &\phantom{\equiv~\{} \cup\, \ptest{\lnot(\pmb{V_1 \cdot f_1-V_2 \cdot P \cdot f_2 < (L_L-x_1)/\scduration} \lor f_2 \leq FL \lor \pmb{V_2 \cdot P \cdot f_2 > (H_H-x_2)/\scduration})}~\}
\end{align*}}
\caption{Safe controller with mixed decision conditions for valve and pump actuation based on flow rate and empirical thresholds (replaces the controller of Figure~\ref{fig:generated-hp-program}, only the control decisions that exposed counterexamples in \KeYmaeraX are changed; changes are highlighted in boldface)}
\label{fig:safe-hp-program}
\end{figure*}

\begin{figure}
\centering
\begin{lstlisting}[linewidth=\columnwidth,
%basicstyle=\footnotesize,
language=ST]
  IF (f1 > (HH-x1)/$\scduration$) THEN V1:=0;
  ELSE 
    IF (x1 <= L1) THEN V1:=1; END_IF;
  END_IF;
  
  IF (x2 <= L2) THEN P:=1; V2:=1; END_IF;
  
  IF (V1*f1-V2*P*f2 < (LL-x1)/$\scduration$ OR f2<=FL OR V2*P*f2 > (HH-x2)/$\scduration$) THEN 
    P:=0; V2:=0; 
  END_IF;
\end{lstlisting}
\caption{ST code fragment compiled from safe \(\text{ctrl}\) (see Figure~\ref{fig:safe-hp-program}). The variable \(\scduration\) is a placeholder for the concrete task interval time}
\label{fig:safe-st-program}
\end{figure}

\subsubsection{Generating Safe PLC Code}
The hybrid program was updated to reflect a safe system that restricts the flow rates by modifying the guard values on the discrete control. Figure~\ref{fig:safe-hp-program} shows the updated hybrid program that was proven to be safe with \kyx. 
Once verified, \name generates the associated PLC code, listed in Figure~\ref{fig:safe-st-program}.

\begin{figure*}[tbp]
	\centering
    \begin{subfigure}[b]{0.42\textwidth}
    	\centering
    	\includegraphics[width=1\textwidth]{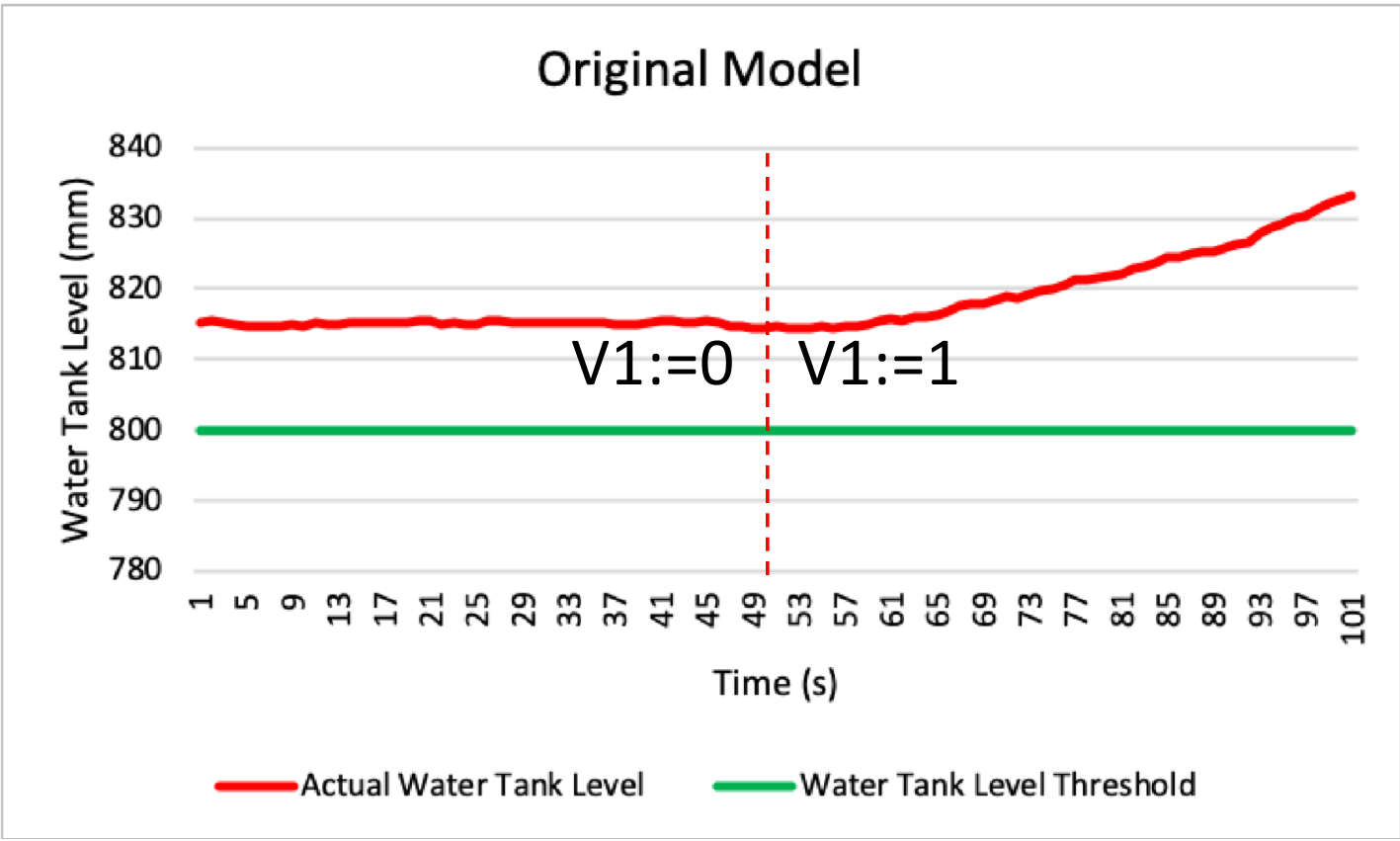}
        \caption{Original Model}
	\end{subfigure}
~
	\begin{subfigure}[b]{0.42\textwidth}
    	\centering
    	\includegraphics[width=1\textwidth]{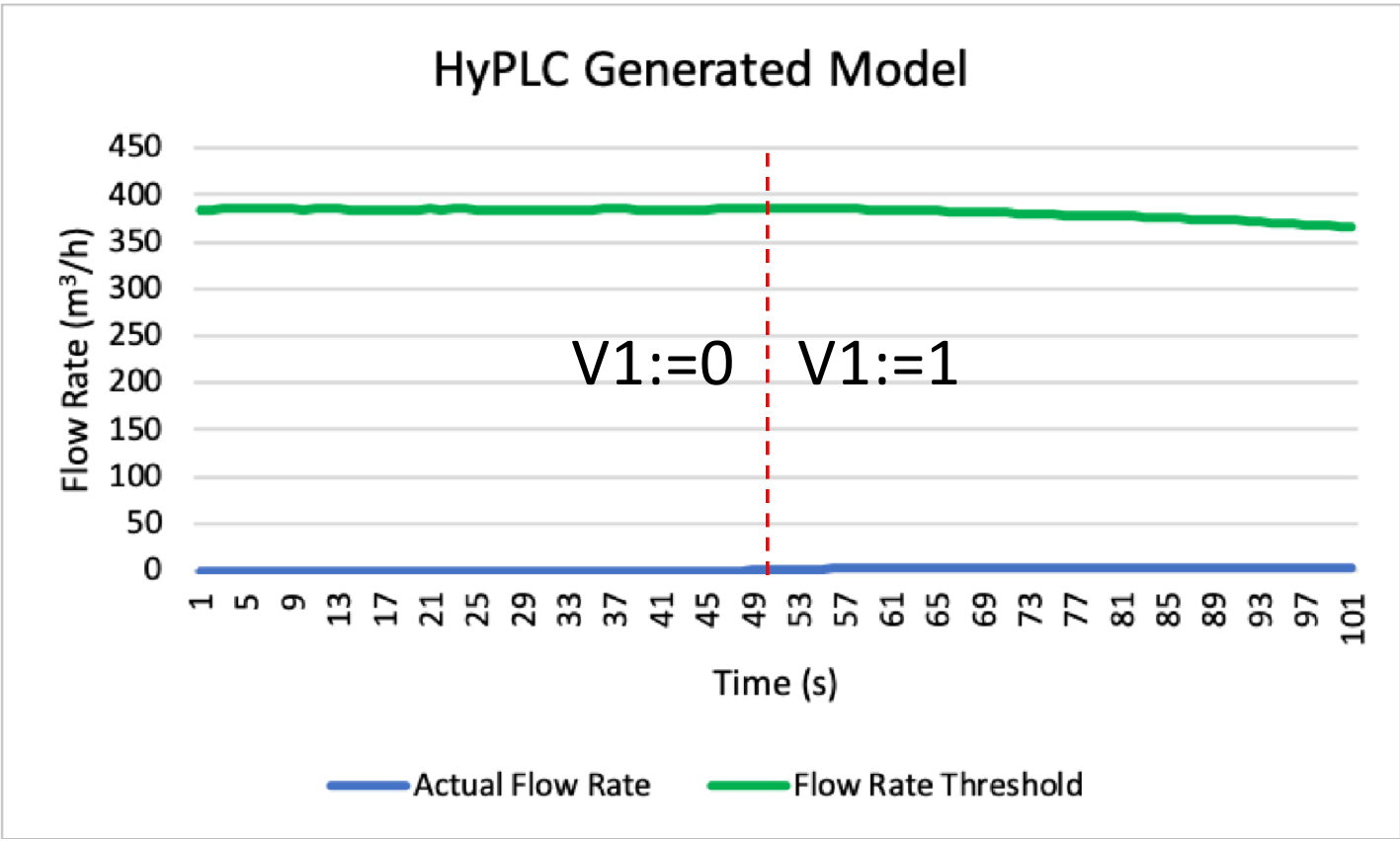}
        \vspace{-0.15in}
        \caption{\name Generated Model}
	\end{subfigure}
	\caption{Comparison of original model versus \name generated model on the same snippet of real data. The dashed red line shows when valve V1 is opened. For the original model, a flag would be raised as it violates the associated incorrect safety constraints--as indicated by the solid red line. However, the \name generated model would not have raised a false safety violation flag--as indicated by the blue line.\andre{different colors and quantities confusing. Plot all?}\luis{I have tried to plot all in the same figure but it doesn't visualize well. It doesn't help that the actual flow rate is so much smaller than the threshold.}}
	\label{fig:comparison-real-data}
\end{figure*}

\noindent\textbf{Comparison on real-world data.} To illustrate the safety guarantees of our system, we developed a Python script to analyze the sensor and actuation values of 4 days worth of sensor data~\cite{gohAdepuJunejoMathur}. 
We check the values of the sensor data relevant to the process described by our model and instantiate the parameters in the model with the values provided in the dataset. 
At each time sample, the script checks that the collective system state complies\footnote{The relevant conditions to check and expected control choices can be extracted by proof from a hybrid program using ModelPlex~\cite{DBLP:journals/fmsd/MitschP16}.} with the expected test-actuation sequences enumerated in our model: the recorded actuator commands for the valves and pump must match the expected command from our model, which is determined by matching the recorded sensor values with the test conditions in the model.
For instance, the script records a violation if the condition in
\lstinline[language=ST]{IF (f1 > (HH-x1)/$\scduration$) THEN V1:=0;} 
from Figure~\ref{fig:safe-st-program} is met but the recorded actuation differs from closing the valve.
We compared the violations with the original ST program in Figure~\ref{fig:swat-simple-st-code} where, e.g., the corresponding condition reads \lstinline[language=ST]{IF (x1 >= H1) THEN V1:=0;}. Figure~\ref{fig:comparison-real-data} illustrates this difference on a snippet of the real data. In particular, this snippet represents a period of manual operation. The original model will have raised an incorrect safety violation flag, while the model generated by \name will not raise a flag during this period.

Our results revealed that the recorded data did not comply with Figure~\ref{fig:safe-st-program} for 238 instances\footnote{An instance of a model compliance violation is a range of uninterrupted scan cycles where the recorded data deviates from the expected model.} for the verified code in Figure~\ref{fig:safe-st-program} and 439 instances for the original code in Figure~\ref{fig:swat-simple-st-code} out of ~40K possible instances\footnote{For 403K samples, the duration of each instance was on average 10 scan cycles.}.  Note that the verified code  allows the system to operate closer to its limits for reasons detailed below, providing a more efficient system operation while enjoying the safety guarantees of the proofs in \KeYmaeraX. 

Upon inspection, most of the violations observed occur during initialization and at the thresholds in oscillating normal system operation~\cite{gohAdepuJunejoMathur}. 
For example, during initialization, the data shows a period where valve $V_1$ is closed and the tank is drained despite not having reached the lower threshold \(L_1\), see \cite[Fig. 4b]{gohAdepuJunejoMathur}. 
During normal operation, the system slightly overshoots or undershoots the intended limits for discrete switching states, e.g., if the system was supposed to close $V_1$ when $x_1 \leq L_1$, the system may undershoot $L_1$.
These slight overshoots or undershoots are not allowed in the original ST code, but can be tolerated in the verified model that takes into account flow rates for making decisions.

This study allowed us to not only generate safe PLC code, but to also reveal missing conditions in PLC code that has been evaluated empirically to be safe. We further showed that \name may provide a means of operating a system closer to safety limits while at the same time \emph{provably maintaining crucial safety guarantees}.

\section{Conclusion and Future Work}\label{sec:conclusion}
In this paper, we formalize compilation between safety-critical code utilized in industrial control systems (ICS) and the discrete control of hybrid programs specified in differential dynamic logic (\dl{}). We present \name, a tool for bi-directional compilation of code loaded onto programmable logic controllers (PLCs) to and from hybrid programs specified in \dl{} to provide safety guarantees for ``deep'' correctness properties of the PLC code in the context of the cyber-physical ICS. We evaluated \name on a real water treatment testbed, demonstrating how \name can be utilized to both verify the safety of existing PLC code as well as generate correct PLC code given a verified hybrid program. Future work will focus on lifting assumptions for PLC arithmetic, support for multiple tasks, as well as support for security analysis. This work serves as a foundation for pragmatic verification of PLC code as well as to understand the safety implications of a particular implementation given complex cyber-physical interdependencies.
\section*{Acknowledgment}
We would like to thank the U.S. Department of Education (Graduate Assistance in Areas of National Need), the U.S. Department of Energy (Award Number DE-OE0000780), the Air Force Office of Scientific Research (Award Number FA9550-16-1-0288), as well as the Defense Advanced Research Projects Agency (Award Number FA8750-18-C-0092) for their support of this work. The views and conclusions contained in this document are those of the authors and should not be interpreted as representing the official policies, either expressed or implied, of the U.S. Department of Education, the U.S. Department of Energy, AFOSR, DARPA, or the U.S. Government. The U.S. Government is authorized to reproduce and distribute reprints for Government purposes
notwithstanding any copyright notation here on.
\input

\bibliographystyle{IEEEtran}
\bibliography{references}

\begin{thebibliography}{10}
\providecommand{\url}[1]{#1}
\csname url@samestyle\endcsname
\providecommand{\newblock}{\relax}
\providecommand{\bibinfo}[2]{#2}
\providecommand{\BIBentrySTDinterwordspacing}{\spaceskip=0pt\relax}
\providecommand{\BIBentryALTinterwordstretchfactor}{4}
\providecommand{\BIBentryALTinterwordspacing}{\spaceskip=\fontdimen2\font plus
\BIBentryALTinterwordstretchfactor\fontdimen3\font minus
  \fontdimen4\font\relax}
\providecommand{\BIBforeignlanguage}[2]{{%
\expandafter\ifx\csname l@#1\endcsname\relax
\typeout{** WARNING: IEEEtran.bst: No hyphenation pattern has been}%
\typeout{** loaded for the language `#1'. Using the pattern for}%
\typeout{** the default language instead.}%
\else
\language=\csname l@#1\endcsname
\fi
#2}}
\providecommand{\BIBdecl}{\relax}
\BIBdecl

\bibitem{mcgranaghan1993voltage}
M.~F. McGranaghan, D.~R. Mueller, and M.~J. Samotyj, ``Voltage sags in
  industrial systems,'' \emph{IEEE Transactions on industry applications},
  vol.~29, no.~2, pp. 397--403, 1993.

\bibitem{abbpluto2016}
\BIBentryALTinterwordspacing
``{ABB} launches new {Pluto} programmable logic controller for rail safety
  applications.'' [Online]. Available: \url{http://www.abb.com/cawp/seitp202/
  fa405fb9803dd9eac1258035002f53c0.aspx}
\BIBentrySTDinterwordspacing

\bibitem{kesler2011vulnerability}
B.~Kesler, ``The vulnerability of nuclear facilities to cyber attack; strategic
  insights: Spring 2010,'' \emph{Strategic Insights, Monterey, California.
  Naval Postgraduate School, Spring 2011}, 2011.

\bibitem{manesis1998intelligent}
S.~Manesis, D.~Sapidis, and R.~King, ``Intelligent control of wastewater
  treatment plants,'' \emph{Artificial Intelligence in Engineering}, vol.~12,
  no.~3, pp. 275--281, 1998.

\bibitem{moon1994modeling}
I.~Moon, ``Modeling programmable logic controllers for logic verification,''
  \emph{IEEE Control Systems}, vol.~14, no.~2, pp. 53--59, 1994.

\bibitem{darvas2015formal}
D.~Darvas, E.~Blanco~Vinuela, and I.~Majzik, ``A formal specification method
  for {PLC}-based applications,'' in \emph{15th International Conference on
  Accelerator and Large Experimental Physics Control Systems}.\hskip 1em plus
  0.5em minus 0.4em\relax {JACoW}, 2015, pp. 907--910.

\bibitem{mader1999timed}
A.~Mader and H.~Wupper, ``Timed automaton models for simple programmable logic
  controllers,'' in \emph{Real-Time Systems, 1999. Proceedings of the 11th
  Euromicro Conference on}.\hskip 1em plus 0.5em minus 0.4em\relax IEEE, 1999,
  pp. 106--113.

\bibitem{thapa2005transformation}
D.~Thapa, S.~Dangol, and G.-N. Wang, ``Transformation from {Petri} nets model
  to programmable logic controller using one-to-one mapping technique,'' in
  \emph{Computational Intelligence for Modelling, Control and Automation, 2005
  and International Conference on Intelligent Agents, Web Technologies and
  Internet Commerce, International Conference on}, vol.~2.\hskip 1em plus 0.5em
  minus 0.4em\relax IEEE, 2005, pp. 228--233.

\bibitem{gerth1995simple}
R.~Gerth, D.~Peled, M.~Y. Vardi, and P.~Wolper, ``Simple on-the-fly automatic
  verification of linear temporal logic,'' in \emph{Protocol Specification,
  Testing and Verification XV}.\hskip 1em plus 0.5em minus 0.4em\relax
  Springer, 1995, pp. 3--18.

\bibitem{clarke1986automatic}
E.~M. Clarke, E.~A. Emerson, and A.~P. Sistla, ``Automatic verification of
  finite-state concurrent systems using temporal logic specifications,''
  \emph{ACM Transactions on Programming Languages and Systems (TOPLAS)},
  vol.~8, no.~2, pp. 244--263, 1986.

\bibitem{fulton2015keymaera}
N.~Fulton, S.~Mitsch, J.-D. Quesel, M.~V{\"o}lp, and A.~Platzer, ``{KeYmaera
  X}: an axiomatic tactical theorem prover for hybrid systems,'' in
  \emph{International Conference on Automated Deduction}.\hskip 1em plus 0.5em
  minus 0.4em\relax Springer, 2015, pp. 527--538.

\bibitem{DBLP:journals/jar/Platzer08}
A.~Platzer, ``Differential dynamic logic for hybrid systems.'' \emph{J. Autom.
  Reas.}, vol.~41, no.~2, pp. 143--189, 2008.

\bibitem{DBLP:journals/jar/Platzer17}
------, ``A complete uniform substitution calculus for differential dynamic
  logic,'' \emph{J. Autom. Reas.}, vol.~59, no.~2, pp. 219--265, 2017.

\bibitem{john2010iec}
K.-H. John and M.~Tiegelkamp, \emph{{IEC} 61131-3: programming industrial
  automation systems: concepts and programming languages, requirements for
  programming systems, decision-making aids}.\hskip 1em plus 0.5em minus
  0.4em\relax Springer Science \& Business Media, 2010.

\bibitem{sousamatiec}
\BIBentryALTinterwordspacing
M.~d. Sousa, ``{MATIEC-IEC} 61131-3 compiler,'' 2014. [Online]. Available:
  \url{https://bitbucket.org/mjsousa/matiec}
\BIBentrySTDinterwordspacing

\bibitem{mathur2016swat}
A.~P. Mathur and N.~O. Tippenhauer, ``{SWaT}: A water treatment testbed for
  research and training on ics security,'' in \emph{Cyber-physical Systems for
  Smart Water Networks (CySWater), 2016 International Workshop on}.\hskip 1em
  plus 0.5em minus 0.4em\relax IEEE, 2016, pp. 31--36.

\bibitem{darvas2017plc}
D.~Darvas, I.~Majzik, and E.~B. Vi{\~n}uela, ``{PLC} program translation for
  verification purposes,'' \emph{Periodica Polytechnica. Electrical Engineering
  and Computer Science}, vol.~61, no.~2, p. 151, 2017.

\bibitem{rausch1998formal}
M.~Rausch and B.~H. Krogh, ``Formal verification of {PLC} programs,'' in
  \emph{American Control Conference, 1998. Proceedings of the 1998},
  vol.~1.\hskip 1em plus 0.5em minus 0.4em\relax IEEE, 1998, pp. 234--238.

\bibitem{mclaughlin2014trusted}
S.~E. McLaughlin, S.~A. Zonouz, D.~J. Pohly, and P.~D. McDaniel, ``A trusted
  safety verifier for process controller code.'' in \emph{NDSS}, vol.~14, 2014.

\bibitem{pavlovic2007automated}
O.~Pavlovic, R.~Pinger, and M.~Kollmann, ``Automated formal verification of
  {PLC} programs written in {IL},'' in \emph{Conference on Automated Deduction
  (CADE)}, 2007, pp. 152--163.

\bibitem{mertke2001formal}
T.~Mertke and G.~Frey, ``Formal verification of {PLC} programs generated from
  signal interpreted {Petri} nets,'' in \emph{Systems, Man, and Cybernetics,
  2001 IEEE International Conference on}, vol.~4.\hskip 1em plus 0.5em minus
  0.4em\relax IEEE, 2001, pp. 2700--2705.

\bibitem{darvas2015plcverif}
D.~Darvas, E.~Blanco~Vinuela, and B.~Fern{\'a}ndez~Adiego, ``{PLCverif}: A tool
  to verify {PLC} programs based on model checking techniques,'' in \emph{15th
  International Conference on Accelerator and Large Experimental Physics
  Control Systems}.\hskip 1em plus 0.5em minus 0.4em\relax {JACoW}, 2015, pp.
  911--914.

\bibitem{tapken1998moby}
J.~Tapken and H.~Dierks, ``{MOBY}/{PLC}--graphical development of
  {PLC}-automata,'' in \emph{International Symposium on Formal Techniques in
  Real-Time and Fault-Tolerant Systems}.\hskip 1em plus 0.5em minus 0.4em\relax
  Springer, 1998, pp. 311--314.

\bibitem{sacha2005automatic}
K.~Sacha, ``Automatic code generation for {PLC} controllers,'' in
  \emph{International Conference on Computer Safety, Reliability, and
  Security}.\hskip 1em plus 0.5em minus 0.4em\relax Springer, 2005, pp.
  303--316.

\bibitem{darvas2016conformance}
D.~Darvas, I.~Majzik, and E.~B. Vi{\~n}uela, ``Conformance checking for
  programmable logic controller programs and specifications,'' in
  \emph{Industrial Embedded Systems (SIES), 2016 11th IEEE Symposium on}.\hskip
  1em plus 0.5em minus 0.4em\relax IEEE, 2016, pp. 1--8.

\bibitem{flordal2007automatic}
H.~Flordal, M.~Fabian, K.~{\AA}kesson, and D.~Spensieri, ``Automatic model
  generation and {PLC}-code implementation for interlocking policies in
  industrial robot cells,'' \emph{Control Engineering Practice}, vol.~15,
  no.~11, pp. 1416--1426, 2007.

\bibitem{bohrer2018veriphy}
B.~Bohrer, Y.~K. Tan, S.~Mitsch, M.~O. Myreen, and A.~Platzer, ``{VeriPhy}:
  verified controller executables from verified cyber-physical system models,''
  in \emph{Proceedings of the 39th ACM SIGPLAN Conference on Programming
  Language Design and Implementation}.\hskip 1em plus 0.5em minus 0.4em\relax
  ACM, 2018, pp. 617--630.

\bibitem{majumdar2013compositional}
R.~Majumdar, I.~Saha, K.~Ueda, and H.~Yazarel, ``Compositional equivalence
  checking for models and code of control systems.''\hskip 1em plus 0.5em minus
  0.4em\relax {IEEE}, 12 2013, pp. 1564--1571.

\bibitem{platzer2010logical}
A.~Platzer, \emph{Logical analysis of hybrid systems: proving theorems for
  complex dynamics}.\hskip 1em plus 0.5em minus 0.4em\relax Springer Science \&
  Business Media, 2010.

\bibitem{allenbradley2008}
\BIBentryALTinterwordspacing
{\relax Rockwell Automation}, ``{Logix5000} controllers, tasks, programs, and
  routines,'' 2018. [Online]. Available:
  \url{https://literature.rockwellautomation.com/idc/groups/literature/documents/pm/1756-pm005_-en-p.pdf}
\BIBentrySTDinterwordspacing

\bibitem{gohAdepuJunejoMathur}
J.~Goh, S.~Adepu, K.~N. Junejo, and A.~Mathur, ``{A Dataset to Support Research
  in the Design of Secure Water Treatment Systems},'' in \emph{The 11th
  International Conference on Critical Information Infrastructures Security
  (CRITIS)}.\hskip 1em plus 0.5em minus 0.4em\relax New York, USA: Springer,
  October 2016, pp. 1--13.

\bibitem{DBLP:journals/fmsd/MitschP16}
S.~Mitsch and A.~Platzer, ``{ModelPlex}: Verified runtime validation of
  verified cyber-physical system models,'' \emph{Form. Methods Syst. Des.},
  vol.~49, no.~1, pp. 33--74, 2016, special issue of selected papers from
  RV'14.

\bibitem{DBLP:conf/lics/Platzer12a}
A.~Platzer, ``Logics of dynamical systems,'' in \emph{LICS}.\hskip 1em plus
  0.5em minus 0.4em\relax IEEE, 2012, pp. 13--24.

\end{thebibliography}

\pagebreak

\appendix
\section{Semantics of Differential Dynamic Logic}
\label{sec:dlsemantics}

The semantics of \dL~\cite{DBLP:journals/jar/Platzer08,DBLP:conf/lics/Platzer12a,DBLP:journals/jar/Platzer17} is a Kripke semantics in which the states of the Kripke model are the states of the hybrid system.
Let \(\reals\) denote the set of real numbers and \(\allvars\) denote the set of variables.
A state is a map~\(\iget[state]{\I}: \allvars\to\reals\) assigning a real value \(\iget[state]{\I}(x)\) to each variable \(x\in\allvars\).
We write~\m{\imodels{\I}{\phi}} if formula \(\phi\) is true at \(\iname[state]{\I}~\iget[state]{\I}\)~(Def.~\ref{def:valuation}).
The real value of term \(\theta\) at \(\iget[state]{\I}\) is denoted \(\ivaluation{\I}{\theta}\).

The semantics of translatable hybrid programs \(\alpha\) is expressed as a transition relation between states (Def.~\ref{def:valuationProgram}).

\begin{definition}[Transition semantics of translatable hybrid programs] \label{def:valuationProgram}
The transition relation \(\iaccess[\alpha]{\I}\) specifies which \(\iname[state]{\I}\)s \(\wt\) are reachable from a \(\iname[state]{\I}\) \(\ws\) by operations of \(\alpha\).
It is defined as follows:
    \begin{enumerate}
    \item
      \(\relateds{\iaccess[\pupdate{\umod{x}{\theta}}]{\I}}{\ws}{\wt}\)
      iff \(\wt(x)=\ivaluation{\I}{\theta}\), and for all other variables \(z \neq x\), \(\wt(z) = \ws(z)\)
    \item
      \(\relateds{\iaccess[\ptest{\ivr}]{\I}}{\ws}{\wt}\)
      iff~$\iget[state]{\I}=\wt$ and~\m{\imodels{\I}{\ivr}}
    \item \(\iaccess[\pchoice{\alpha}{\beta}]{\I} =
      \iaccess[\alpha]{\I} \cup \iaccess[\beta]{\I}\)
      \index{$\pchoice{}{}$}
    \item \(\iaccess[{\alpha};{\beta}]{\I} =
      \{(\ws,\wt) ~|~ \textrm{exists}~\mu~\textrm{s.t.}~ \related{\iaccess[\alpha]{\I}}{\ws}{\mu} ~\textrm{and}~ \related{\iaccess[\beta]{\I}}{\mu}{\wt}\}\)
    \end{enumerate}
\end{definition}

{%
\begin{definition}[Interpretation of translatable \dL formulas] \label{def:valuation}
  Truth of translatable \dL formula \(\phi\) in \(\iname[state]{\I}~\iget[state]{\I}\), written \(\imodels{\I}{\phi}\), is defined as follows:
  \begin{enumerate}
  \item \(\imodels{\I}{\theta_1\sim\theta_2}\)
    iff \(\ivaluation{\I}{\theta_1} \sim \ivaluation{\I}{\theta_2}\)
    for ${\sim} \in \{=,\leq,<,\geq,>\}$
  \item \(\imodels{\I}{\phi \land \psi}\) iff
    \(\imodels{\I}{\phi}\) and \(\imodels{\I}{\psi}\),
    so on for $\lnot,\lor,\limply,\lbisubjunct$
  \item
    \(\imodels{\I}{\dbox{\alpha}{\phi}}\)
      iff
      \(\imodels{\It}{\phi}\)
      for all \({\wt ~\textrm{with}~
        \related{\iaccess[\alpha]{\I}}{\ws}{\wt}}\)
%
  \end{enumerate}
 We denote \emph{validity} as \(\models \phi\), i.e., \(\imodels{\I}{\phi}\) for all states $\iget[state]{\I}$.
\end{definition}
}

\section{Semantics of Structured Text}
\label{sec:structured-text-semantics}

\begin{figure}[h]
\begin{align}
 \frac{\iget[state]{\I}(x_1) = c_1}{(x_1,\iget[state]{\I})\stopeval c_1} &~\textrm{ST Variable Value}\nonumber\\
 \frac{(e_1,\iget[state]{\I})\stopeval c_1 \qquad (e_2,\iget[state]{\I})\stopeval c_2}{(e_1 \textrm{~\text{OR}~}  e_2,\iget[state]{\I})\stopeval c_1 \lor c_2} &~\textrm{ST OR Expression}\nonumber\\
 \frac{(e_1,\iget[state]{\I})\stopeval c_1 \qquad  (e_2,\iget[state]{\I})\stopeval c_2}{(e_1 \textrm{~\text{AND}~}  e_2,\iget[state]{\I})\stopeval c_1 \land c_2} &~\textrm{ST AND Expression}\nonumber\\
 \frac{(\stprg{s_1},\iget[state]{\I}) \stprgeval (\stprg{\stskip},\iget[state]{\It})}{(\stprg{s_1;s_2}\,,\,\iget[state]{\I}) \stprgeval (\stprg{s_2},\iget[state]{\It})} &~\textrm{ST Sequence}\nonumber\\
 \frac{(e_1,\iget[state]{\I})\stopeval c_1}{(\stprg{x_1:=e_1}\,,\,\iget[state]{\I}) \stprgeval (\stprg{\stskip},\iget[state]{\I}[x_1 \mapsto c_1])} &~\textrm{ST Assignment}\nonumber\\
 \frac{(e_1,\iget[state]{\I}) \stopeval \top}{(\stif{e_1}{s_1}{s_2}\,,\,\iget[state]{\I}) \stprgeval (s_1,\iget[state]{\I})} &~\textrm{ST IF (1)}\nonumber\\
\frac{(e_1,\iget[state]{\I}) \stopeval \bot}{(\stprg{\stif{e_1}{s_1}{s_2}}\,,\,\iget[state]{\I}) \stprgeval (s_2,\iget[state]{\I})} &~\textrm{ST IF (2)}\nonumber\\
 \frac{(e_1,\iget[state]{\I}) \stopeval \top}{(\stifthen{e_1}{s_1}\,,\,\iget[state]{\I}) \stprgeval (s_1,\iget[state]{\I})} &~\textrm{ST IF (3)}\nonumber\\
\frac{(e_1,\iget[state]{\I}) \stopeval \bot}{(\stprg{\stifthen{e_1}{s_1}}\,,\,\iget[state]{\I}) \stprgeval (\stskip,\iget[state]{\I})} &~\textrm{ST IF (4)}\nonumber
\end{align}
\caption{Operational ST semantics based on~\cite{darvas2017plc}.}
\label{fig:st-semantics}
\end{figure}

Figure~\ref{fig:st-semantics} lists the operational ST semantics based on~\cite{darvas2017plc}. 
The context of an ST statement is denoted by $\iget[state]{\I}$, which is a function $\iget[state]{\I}: V\rightarrow D$ that assigns a value from pre-defined domains \(D\) to each defined variable \(x \in V\) (we use \(D=\text{LReal}\)). 
We follow~\cite{darvas2017plc} where a program $s$ is executed from an initial context $\iget[state]{\I}$ and results in a subsequent context $\iget[state]{\It}$ that determines the values of the physical outputs and values of the variables for the subsequent PLC scan cycle. 
Execution of a cycle ends when the final configuration $(\stprg{\stskip}, \iget[state]{\I})$ is reached, where \(\stskip\) denotes the end of the code for this scan cycle. 
Variables are denoted as $x_1$, expressions are denoted as $e_1$ and $e_2$, constants are denoted as $c_1$ and $c_2$, and ST statements are denoted as $s_1$ and $s_2$. 
Arithmetic/logic evaluations are denoted by $\stopeval$ and single-step program evaluations are denoted as $\stprgeval$. 
The context $\iget[state]{\I}[x_1 \mapsto c_1]$ denotes a context \(\iget[state]{\It}\) that agrees with \(\iget[state]{\I}\) except that \(\iget[state]{\It}(x_1) = c_1\). 
Figure~\ref{fig:st-semantics} lists \text{AND} and \text{OR} as representative examples for the other arithmetic/logical operations for two ST terms as defined by the IEC 61131-3 standard.

\section{Proofs}
\label{sec:proofs}

\begin{proof}[Proof of Lemma~\ref{lemma:terms}]
Straightforward structural induction from 
\begin{itemize}
\item Number literals \((n,\iget[state]{\I}) \stopeval n\) for all \(\iget[state]{\I}\) and \(\ivaluation{}{n}=n\)
\item ST Variable Value \(\frac{\iget[state]{\I}(x)=c}{(x,\iget[state]{\I}) \stopeval c}\) and \dL variable valuation \(\ivaluation{\I}{x} = \iget[state]{\I}(x)\)
\item Negation in ST
\(\frac{(\theta,\iget[state]{\I})=c_1}{(-\theta,\iget[state]{\I})\stopeval -c_1}\) and \dL \(\ivaluation{\I}{-\theta}=-\ivaluation{\I}{\theta}\) for term $\theta$
\item Binary arithmetic operator $\arithmeticOp$ in ST
\(\frac{(\theta,\iget[state]{\I})=c_1 ~~~ (\eta,\iget[state]{\I})=c_2}{(\theta \arithmeticOp \eta,\iget[state]{\I})\stopeval c_1 \arithmeticOp c_2}
\)
and \dL \( \ivaluation{\I}{\theta \arithmeticOp \eta} = \ivaluation{\I}{\theta} \arithmeticOp \ivaluation{\I}{\eta}
\) for terms $\theta$ and $\eta$.
\qedhere
\end{itemize}
\end{proof}

\begin{proof}[Proof of Lemma~\ref{lemma:formula}]
Straightforward structural induction from 
\begin{itemize}
\item Comparisons $\relationalOpST\ \in \{<,>,>=,<=,<>,=\}$ in ST
\[\frac{(\theta,\iget[state]{\I})=c_1 \quad (\eta,\iget[state]{\I})=c_2}{(\theta \relationalOpST \eta,\iget[state]{\I})\stopeval c_1 \relationalOpST c_2}
\]
and \dL \(\ivaluation{\I}{\theta \relationalOpHP \eta} =\ivaluation{\I}{\theta} \relationalOpHP \ivaluation{\I}{\eta}
\) for operator \(\relationalOpHP \in \{<,\leq,=,\geq,>\}\) and terms $\theta$ and $\eta$.
\item Logical connectives in ST \[\frac{(\phi,\iget[state]{\I}) \stopeval p_1 \quad (\psi,\iget[state]{\I}) \stopeval p_2}{(\phi ~\text{AND}~ \psi,\iget[state]{\I}) \stopeval p_1 \land p_2}\enspace;\] in \dL \(\ivaluation{\I}{\phi \land \psi} = \ltrue\) iff \(\iget[state]{\I} \models \phi\) and \(\iget[state]{\I} \models \psi\), accordingly for \(\lnot,\lor\).
\qedhere
\end{itemize}
\end{proof}

\begin{proof}[Proof of Lemma~\ref{lemma:st-to-hp}]
By structural induction on ST programs from the base case \((\stprg{\stskip},\iget[state]{\I})\), so \((\iget[state]{\I},\iget[state]{\I}) \in \iaccess[\ptest{\ltrue}]{\I}\), and induction hypothesis \((\stprg{\stprogrameq{1}},\iget[state]{\I}) \stprgeval (\stprg{\stskip},\iget[state]{\It})\) then \((\iget[state]{\I},\iget[state]{\It}) \in \iaccess[\compilesteq{\stprogrameq{1}}]{\I}\).

\begin{case}[Assignment \(\compilesteq{\stassign{x}{\theta}} \compileop {\humod{\compilesteq{x}}{\compilesteq{\theta}}}\)] \label{case:det-assignment1}
We have to show \((\iget[state]{\I},\iget[state]{\It}) \in \iaccess[\humod{\compilesteq{x}}{\compilesteq{\theta}}]{\I}\).
Direct consequence of ST Assignment and induction hypothesis: \((\stprg{\stassign{x}{\theta}},\iget[state]{\I}) \stprgeval (\stprg{\stskip},\stsubst{\iget[state]{\I}}{x}{c})\) for \(\iget[state]{\I}(x) \stopeval c\), i.e., \(\stsubst{\iget[state]{\I}}{x}{c}\) denotes a state \(\iget[state]{\It}=\iget[state]{\I}\) except \(\iget[state]{\It}(x)\) is replaced with \(\theta\) at \(\iget[state]{\I}\).
\end{case}

\begin{case}[Sequential \(\compilesteq{\stprogrameq{1};\stprogrameq{2}} \compileop \compilesteq{\stprogrameq{1}};\compilesteq{\stprogrameq{2}}\)] \label{case:sequential-composition1}
We have to show \((\iget[state]{\I},\iget[state]{\It}) \in \iaccess[\compilesteq{\stprogrameq{1}};\compilesteq{\stprogrameq{2}}]{\I}\).
Let \((\stprg{\stprogrameq{1};\stprogrameq{2}},\iget[state]{\I}) \stprgeval (\stprg{\stskip},\iget[state]{\It})\) so by operational semantics there is \(\iget[state]{\Isig}\) such that \((\stprg{\stprogrameq{1}},\iget[state]{\I}) \stprgeval (\stprg{\stskip},\iget[state]{\Isig})\) and \((\stprg{\stprogrameq{2}},\iget[state]{\Isig}) \stprgeval (\stprg{\stskip},\iget[state]{\It})\).
From \((\stprg{\stprogrameq{1}},\iget[state]{\I}) \stprgeval (\stprg{\stskip},\iget[state]{\Isig})\) we get \((\iget[state]{\I},\iget[state]{\Isig}) \in \iaccess[\compilesteq{\stprogrameq{1}}]{\I}\) by induction hypothesis, and from \((\stprg{\stprogrameq{2}},\iget[state]{\Isig}) \stprgeval (\stprg{\stskip},\iget[state]{\It})\) we get \((\iget[state]{\Isig},\iget[state]{\It}) \in \iaccess[\compilesteq{\stprogrameq{2}}]{\I}\) by induction hypothesis and conclude \((\iget[state]{\I},\iget[state]{\It}) \in \iaccess[\compilesteq{\stprogrameq{1}};\compilesteq{\stprogrameq{2}}]{\I}\) by \dL.
\end{case}
\begin{case}[If-Then-Else \(\compilesteq{\stif{\ivr}{\stprogrameq{1}}{\stprogrameq{2}}} \compileop\\ \pchoice{\ptest{\compilesteq{\ivr};\compilesteq{\stprogrameq{1}}}}{\ptest{\lnot\compilesteq{\ivr}};\compilesteq{\stprogrameq{2}}}\)] 
Let \((\stif{\ivr}{\stprogrameq{1}}{\stprogrameq{2}},\iget[state]{\I}) \stprgeval (\stprg{\stskip},\iget[state]{\It})\).
We have to show \[(\iget[state]{\I},\iget[state]{\It}) \in \iaccess[\pchoice{\ptest{\compilesteq{\ivr}};\compilesteq{\stprogrameq{1}}}{\ptest{\lnot\compilesteq{\ivr}};\compilesteq{\stprogrameq{2}}}]{\I}\]

(i) Case \(\iget[state]{\I} \models \ivr\) and so \(\iget[state]{\I}(\ivr) \stopeval \top\) and in turn \((\iget[state]{\I},\iget[state]{\I}) \in \iaccess[\ptest{\compilesteq{\ivr}}]{\I}\) by Lemma~\ref{lemma:formula}.
By ST IF (1) \cite{darvas2017plc} therefore \((\stprg{\stif{\ivr}{\stprogrameq{1}}{\stprogrameq{2}}},\iget[state]{\I}) \stprgeval (\stprg{\stprogrameq{1}},\iget[state]{\I})\) and so \((\stprg{\stprogrameq{1}},\iget[state]{\I}) \stprgeval (\stprg{\stskip},\iget[state]{\It})\). 
Hence in turn we get 
\((\iget[state]{\I},\iget[state]{\It}) \in \iaccess[\compilesteq{\stprogrameq{1}}]{\I}\) by induction hypothesis.
Now \((\iget[state]{\I},\iget[state]{\I}) \in \iaccess[\ptest{\compilesteq{\ivr}}]{\I}\) and \((\iget[state]{\I},\iget[state]{\It}) \in \iaccess[\compilesteq{\stprogrameq{1}}]{\I}\) and so we get
\((\iget[state]{\I},\iget[state]{\It}) \in \iaccess[\ptest{\compilesteq{\ivr}};\compilesteq{\stprogrameq{1}}]{\I}\) by \dL.

(ii) Case \(\iget[state]{\I} \not\models \ivr\) and so \(\iget[state]{\I}(\ivr) \stopeval \bot\) and in turn \((\iget[state]{\I},\iget[state]{\I}) \in \iaccess[\ptest{\lnot\compilesteq{\ivr}}]{\I}\) by Lemma~\ref{lemma:formula}.
By ST IF (2) therefore \((\stprg{\stif{\ivr}{\stprogrameq{1}}{\stprogrameq{2}}},\iget[state]{\I}) \stprgeval (\stprg{\stprogrameq{2}},\iget[state]{\I})\) and so \((\stprg{\stprogrameq{2}},\iget[state]{\I}) \stprgeval (\stprg{\stskip},\iget[state]{\It})\). 
Hence in turn we get 
\((\iget[state]{\I},\iget[state]{\It}) \in \iaccess[\compilesteq{\stprogrameq{2}}]{\I}\) by induction hypothesis.
Now \((\iget[state]{\I},\iget[state]{\I}) \in \iaccess[\ptest{\lnot\compilesteq{\ivr}}]{\I}\) and \((\iget[state]{\I},\iget[state]{\It}) \in \iaccess[\compilesteq{\stprogrameq{2}}]{\I}\) and so we get
\((\iget[state]{\I},\iget[state]{\It}) \in \iaccess[\ptest{\lnot\compilesteq{\ivr}};\compilesteq{\stprogrameq{2}}]{\I}\) by \dL.

(iii) Now either \((\iget[state]{\I},\iget[state]{\It}) \in \iaccess[\ptest{\compilesteq{\ivr}};\compilesteq{\stprogrameq{1}}]{\I}\) or otherwise \((\iget[state]{\I},\iget[state]{\It}) \in \iaccess[\ptest{\lnot\compilesteq{\ivr}};\compilesteq{\stprogrameq{2}}]{\I}\) and so we conclude \[(\iget[state]{\I},\iget[state]{\It}) \in \iaccess[\pchoice{\ptest{\compilesteq{\ivr}};\compilesteq{\stprogrameq{1}}}{\ptest{\lnot\compilesteq{\ivr}};\compilesteq{\stprogrameq{2}}}]{\I}\]
\end{case}

\begin{case}[If-Then \(\compilesteq{\stifthen{\ivr}{\stprogrameq{1}}} \compileop\\ \pchoice{\ptest{\compilesteq{\ivr};\compilesteq{\stprogrameq{1}}}}{\ptest{\lnot\compilesteq{\ivr}}}\)] 

Let \((\stifthen{\ivr}{\stprogrameq{1}},\iget[state]{\I}) \stprgeval (\stprg{\stskip},\iget[state]{\It})\).
We have to show \((\iget[state]{\I},\iget[state]{\It}) \in \iaccess[\pchoice{\ptest{\compilesteq{\ivr}};\compilesteq{\stprogrameq{1}}}{\ptest{\lnot\compilesteq{\ivr}}}]{\I}\).

(i) Case \(\iget[state]{\I} \models \ivr\) and so \(\iget[state]{\I}(\ivr) \stopeval \top\) and in turn \((\iget[state]{\I},\iget[state]{\I}) \in \iaccess[\ptest{\compilesteq{\ivr}}]{\I}\) by Lemma~\ref{lemma:formula}.
By ST IF (3) therefore \((\stprg{\stifthen{\ivr}{\stprogrameq{1}}},\iget[state]{\I}) \stprgeval (\stprg{\stprogrameq{1}},\iget[state]{\I})\) and so \((\stprg{\stprogrameq{1}},\iget[state]{\I}) \stprgeval (\stprg{\stskip},\iget[state]{\It})\). 
Hence in turn we get 
\((\iget[state]{\I},\iget[state]{\It}) \in \iaccess[\compilesteq{\stprogrameq{1}}]{\I}\) by induction hypothesis.
Now \((\iget[state]{\I},\iget[state]{\I}) \in \iaccess[\ptest{\compilesteq{\ivr}}]{\I}\) and \((\iget[state]{\I},\iget[state]{\It}) \in \iaccess[\compilesteq{\stprogrameq{1}}]{\I}\) and so we get
\((\iget[state]{\I},\iget[state]{\It}) \in \iaccess[\ptest{\compilesteq{\ivr}};\compilesteq{\stprogrameq{1}}]{\I}\) by \dL.

(ii) Case \(\iget[state]{\I} \not\models \ivr\) and so \(\iget[state]{\I}(\ivr) \stopeval \bot\) and in turn \((\iget[state]{\I},\iget[state]{\I}) \in \iaccess[\ptest{\lnot\compilesteq{\ivr}}]{\I}\) by Lemma~\ref{lemma:formula}.
By ST IF (4) therefore \((\stprg{\stifthen{\ivr}{\stprogrameq{1}}},\iget[state]{\I}) \stprgeval (\stprg{\stskip},\iget[state]{\I})\). 
Now \((\iget[state]{\I},\iget[state]{\I}) \in \iaccess[\ptest{\lnot\compilesteq{\ivr}}]{\I}\) and \(\iget[state]{\I}=\iget[state]{\It}\) and so we get
\((\iget[state]{\I},\iget[state]{\It}) \in \iaccess[\ptest{\lnot\compilesteq{\ivr}}]{\I}\) by \dL.

(iii) Now either \((\iget[state]{\I},\iget[state]{\It}) \in \iaccess[\ptest{\compilesteq{\ivr}};\compilesteq{\stprogrameq{1}}]{\I}\) or \((\iget[state]{\I},\iget[state]{\It}) \in \iaccess[\ptest{\lnot\compilesteq{\ivr}}]{\I}\) and so we conclude \((\iget[state]{\I},\iget[state]{\It}) \in \iaccess[\pchoice{\ptest{\compilesteq{\ivr}};\compilesteq{\stprogrameq{1}}}{\ptest{\lnot\compilesteq{\ivr}}}]{\I}\).
\qedhere
\end{case}
\end{proof}

\begin{proof}[Proof of Lemma~\ref{lemma:hp-to-st}]
By structural induction over hybrid programs from the base case \((\stprg{\stskip},\iget[state]{\I})\) so \((\iget[state]{\I},\iget[state]{\I}) \in \iaccess[\ptest{\ltrue}]{\I}\), and induction hypothesis \((\stprg{\compilehpeq{\alpha}},\iget[state]{\I}) \stprgeval (\stprg{\stskip},\iget[state]{\It})\) then \((\iget[state]{\I},\iget[state]{\It}) \in \iaccess[\alpha]{\I}\).
\setcounter{case}{0}
\begin{case}[Assignment \(\compilehpeq{\humod{x}{\theta}} \compileop \stassign{\compilehpeq{x}}{ \compilehpeq{\theta}}\)]
We have to show \((\iget[state]{\I},\iget[state]{\It}) \in \iaccess[\humod{x}{\theta}]{\I}\).
From ST Assignment, we know that \((\stprg{\stassign{\compilehpeq{x}}{ \compilehpeq{\theta}}},\iget[state]{\I}) \stprgeval (\stprg{\stskip},\stsubst{\iget[state]{\I}}{\compilehpeq{x}}{c})\) for \(\I\compilehpeq{\theta} \stopeval c\), i.e., there is a state \(\iget[state]{\It}\) which agrees with \(\iget[state]{\I}\) except for the value of \(\compilehpeq{x}\): \(\iget[state]{\It}(\compilehpeq{x}) = \iaccess[\compilehpeq{\theta}]{\I}_{\iget[state]{\I}}\) and 
we conclude \((\iget[state]{\I},\iget[state]{\It}) \in \iaccess[\humod{x}{\theta}]{\I}\) by \dL from Lemma \ref{lemma:terms}.
\end{case}
\begin{case}[Sequential \(\compilehpeq{\alpha;\beta} \compileop \compilehpeq{\alpha};\compilehpeq{\beta}\)] \label{case:sequential-composition2}
We have to show \((\iget[state]{\I},\iget[state]{\It}) \in \iaccess[\alpha;\beta]{\I}\).
Let \((\stprg{\compilehpeq{\alpha;\beta}},\iget[state]{\I}) \stprgeval (\stprg{\stskip},\iget[state]{\It})\) and so by operational semantics there is \(\iget[state]{\Isig}\) such that \((\stprg{\compilehpeq{\alpha}},\iget[state]{\I}) \stprgeval (\stprg{\stskip},\iget[state]{\Isig})\) and \((\stprg{\compilehpeq{\beta}},\iget[state]{\Isig}) \stprgeval (\stprg{\stskip},\iget[state]{\It})\).
Therefore in turn \((\iget[state]{\I},\iget[state]{\Isig}) \in \iaccess[\alpha]{\I}\) and \((\iget[state]{\Isig},\iget[state]{\It}) \in \iaccess[\beta]{\Isig}\) by induction hypothesis and we conclude \((\iget[state]{\I},\iget[state]{\It}) \in \iaccess[\alpha;\beta]{\I}\) by \dL.
\end{case}
\begin{case}[Guarded Choice \(\compilehpeq{\pchoice{\ptest{\ivr};\alpha}{\beta}} \compileop \stif{\compilehpeq{\ivr}}{\compilehpeq{\alpha}\\}{\compilehpeq{\beta}}\)]
Let \((\stprg{\stif{\compilehpeq{\ivr}}{\compilehpeq{\alpha}}{\compilehpeq{\beta}}},\iget[state]{\I}) \stprgeval (\stprg{\stskip},\iget[state]{\It}))\).
We have to show \((\iget[state]{\I},\iget[state]{\It}) \in \iaccess[\pchoice{\ptest{\ivr};\alpha}{\beta}]{\I}\).

(i) Case \(\iget[state]{\I} \models \ivr\) and so \(\iget[state]{\I}(\compilehpeq{\ivr}) \stopeval \top\) and in turn \((\iget[state]{\I},\iget[state]{\I}) \in \iaccess[\ptest{\ivr}]{\I}\) by Lemma~\ref{lemma:formula}.
By ST IF (1) therefore \((\stprg{\stif{\compilehpeq{\ivr}}{\compilehpeq{\alpha}}{\compilehpeq{\beta}}},\iget[state]{\I}) \stprgeval (\stprg{\compilehpeq{\alpha}},\iget[state]{\I}))\) and in turn \((\stprg{\compilehpeq{\alpha}},\iget[state]{\I})) \stprgeval (\stprg{\stskip},\iget[state]{\It})\).
Now \((\iget[state]{\I},\iget[state]{\It}) \in \iaccess[\alpha]{\I}\) by induction hypothesis and we therefore get \((\iget[state]{\I},\iget[state]{\It}) \in \iaccess[\ptest{\ivr};\alpha]{\I}\) by \dL.

(ii) Case \(\iget[state]{\I} \not\models \ivr\) and so \(\iget[state]{\I}(\compilehpeq{\ivr}) \stopeval \bot\).
By ST IF (2) therefore \((\stprg{\stif{\compilehpeq{\ivr}}{\compilehpeq{\alpha}}{\compilehpeq{\beta}}},\iget[state]{\I}) \stprgeval (\stprg{\compilehpeq{\beta}},\iget[state]{\I}))\) and so \((\stprg{\compilehpeq{\beta}},\iget[state]{\I})) \stprgeval (\stprg{\stskip},\iget[state]{\It})\) and in turn \((\iget[state]{\I},\iget[state]{\It}) \in \iaccess[\beta]{\I}\) by induction hypothesis.

(iii) Now either \((\iget[state]{\I},\iget[state]{\It}) \in \iaccess[\ptest{\ivr};\alpha]{\I}\) or \((\iget[state]{\I},\iget[state]{\It}) \in \iaccess[\beta]{\I}\) and we conclude \((\iget[state]{\I},\iget[state]{\It}) \in \iaccess[\pchoice{\ptest{\ivr};\alpha}{\beta}]{\I}\) by \dL.
\end{case}

\begin{case}[If-Then-Else \(\compilehpeq{\pchoice{\ptest{\ivr};\alpha}{\ptest{\lnot\ivr};\beta}} \compileop \stif{\compilehpeq{\ivr}}{\\\compilehpeq{\alpha}}{\compilehpeq{\beta}}\)]
We have to show \((\iget[state]{\I},\iget[state]{\It}) \in \iaccess[\pchoice{\ptest{\ivr};\alpha}{\ptest{\lnot\ivr};\beta}]{\I}\).
Let \((\stprg{\stif{\compilehpeq{\ivr}}{\compilehpeq{\alpha}}{\compilehpeq{\beta}}},\iget[state]{\I}) \stprgeval (\stprg{\stskip},\iget[state]{\It}))\).

(i) Case \(\iget[state]{\I} \models \ivr\) and so \(\iget[state]{\I}(\compilehpeq{\ivr}) \stopeval \top\) and in turn \((\iget[state]{\I},\iget[state]{\I}) \in \iaccess[\ptest{\ivr}]{\I}\) by Lemma~\ref{lemma:formula}.
By ST IF (1) therefore \((\stprg{\stif{\compilehpeq{\ivr}}{\compilehpeq{\alpha}}{\compilehpeq{\beta}}},\iget[state]{\I}) \stprgeval (\stprg{\compilehpeq{\alpha}},\iget[state]{\I}))\) and in turn \((\stprg{\compilehpeq{\alpha}},\iget[state]{\I})) \stprgeval (\stprg{\stskip},\iget[state]{\It})\).
Now \((\iget[state]{\I},\iget[state]{\It}) \in \iaccess[\alpha]{\I}\) by induction hypothesis and we therefore get \((\iget[state]{\I},\iget[state]{\It}) \in \iaccess[\ptest{\ivr};\alpha]{\I}\) by \dL.

(ii) Case \(\iget[state]{\I} \not\models \ivr\) and so \(\iget[state]{\I}(\compilehpeq{\ivr}) \stopeval \bot\) and in turn \((\iget[state]{\I},\iget[state]{\I}) \in \iaccess[\ptest{\lnot\ivr}]{\I}\) by Lemma~\ref{lemma:formula}.
By ST IF (2) hence \((\stprg{\stif{\compilehpeq{\ivr}}{\compilehpeq{\alpha}}{\compilehpeq{\beta}}},\iget[state]{\I}) \stprgeval (\stprg{\compilehpeq{\beta}},\iget[state]{\I}))\), so \((\stprg{\compilehpeq{\beta}},\iget[state]{\I})) \stprgeval (\stprg{\stskip},\iget[state]{\It})\) and in turn \((\iget[state]{\I},\iget[state]{\It}) \in \iaccess[\beta]{\I}\) by induction hypothesis; we conclude \((\iget[state]{\I},\iget[state]{\It}) \in \iaccess[\ptest{\lnot\ivr};\beta]{\I}\).

(iii) Now either \((\iget[state]{\I},\iget[state]{\It}) \in \iaccess[\ptest{\ivr};\alpha]{\I}\) or \((\iget[state]{\I},\iget[state]{\It}) \in \iaccess[\ptest{\lnot\ivr};\beta]{\I}\) and we conclude \((\iget[state]{\I},\iget[state]{\It}) \in \iaccess[\pchoice{\ptest{\ivr};\alpha}{\ptest{\lnot\ivr};\beta}]{\I}\) by \dL.
\end{case}

\begin{case}[If-Then \(\compilehpeq{\pchoice{\ptest{\ivr};\alpha}{\ptest{\lnot\ivr}}} \compileop \stifthen{\compilehpeq{\ivr}}{\\\compilehpeq{\alpha}}\)]
Let \((\stprg{\stifthen{\compilehpeq{\ivr}}{\compilehpeq{\alpha}}},\iget[state]{\I}) \stprgeval (\stprg{\stskip},\iget[state]{\It}))\).
We have to show \((\iget[state]{\I},\iget[state]{\It}) \in \iaccess[\pchoice{\ptest{\ivr};\alpha}{\ptest{\lnot\ivr}}]{\I}\).

(i) Case \(\iget[state]{\I} \models \ivr\) and so \(\iget[state]{\I}(\compilehpeq{\ivr}) \stopeval \top\) and in turn \((\iget[state]{\I},\iget[state]{\I}) \in \iaccess[\ptest{\ivr}]{\I}\) by Lemma~\ref{lemma:formula}.
By ST IF (3) therefore \((\stprg{\stifthen{\compilehpeq{\ivr}}{\compilehpeq{\alpha}}},\iget[state]{\I}) \stprgeval (\stprg{\compilehpeq{\alpha}},\iget[state]{\I}))\) and in turn \((\stprg{\compilehpeq{\alpha}},\iget[state]{\I})) \stprgeval (\stprg{\stskip},\iget[state]{\It})\).
Now \((\iget[state]{\I},\iget[state]{\It}) \in \iaccess[\alpha]{\I}\) by induction hypothesis and we therefore get \((\iget[state]{\I},\iget[state]{\It}) \in \iaccess[\ptest{\ivr};\alpha]{\I}\) by \dL.

(ii) Case \(\iget[state]{\I} \not\models \ivr\) and so \(\iget[state]{\I}(\compilehpeq{\ivr}) \stopeval \bot\) and in turn \((\iget[state]{\I},\iget[state]{\I}) \in \iaccess[\ptest{\lnot\ivr}]{\I}\) by Lemma~\ref{lemma:formula}.
By ST IF (4) therefore \((\stprg{\stifthen{\compilehpeq{\ivr}}{\compilehpeq{\alpha}}},\iget[state]{\I}) \stprgeval (\stprg{\stskip},\iget[state]{\I}))\) and so \(\iget[state]{\I}=\iget[state]{\It}\) and in turn \((\iget[state]{\I},\iget[state]{\It}) \in \iaccess[\ptest{\lnot\ivr}]{\I}\).

(iii) Now either \((\iget[state]{\I},\iget[state]{\It}) \in \iaccess[\ptest{\ivr};\alpha]{\I}\) or \((\iget[state]{\I},\iget[state]{\It}) \in \iaccess[\ptest{\lnot\ivr}]{\I}\) and we conclude \((\iget[state]{\I},\iget[state]{\It}) \in \iaccess[\pchoice{\ptest{\ivr};\alpha}{\ptest{\lnot\ivr}}]{\I}\) by \dL.
\qedhere
\end{case}

\end{proof}

\end{document}